\documentclass[11pt]{article}

\usepackage{amssymb,amsmath,amsfonts,amssymb}
\usepackage{graphics,graphicx,color}
\usepackage{empheq}

\def\@abssec#1{\vspace{.05in}\footnotesize \parindent .2in
{\bf #1. }\ignorespaces}

\graphicspath{{/EPSF/}{Figures/}}
\DeclareGraphicsExtensions{.eps}

\newtheorem{theorem}{Theorem}[section]
\newtheorem{lemma}[theorem]{Lemma}
\newtheorem{proposition}[theorem]{Proposition}
\newtheorem{corollary}[theorem]{Corollary}
\newtheorem{definition}[theorem]{Definition}
\newtheorem{remark}[theorem]{Remark}

\def \Rm {{\mathbb R}}
\def \Mm {\mathbb M}
\def \Nm {\mathbb N}
\def \Cm {\mathbb C}
\def \Zm {\mathbb Z}
\def \Sm {\mathbb S}

\newcommand{\eps}{\varepsilon}

\newcommand{\dsum}{\displaystyle\sum}
\newcommand{\dint}{\displaystyle\int}

\newcommand{\aver}[1]{\langle {#1} \rangle}

\newcommand{\bpar}{\bar\partial}

\newcommand{\bx}{\mathbf x}

\newcommand{\mF}{\mathcal F}

\newcommand{\fa}{{\mathfrak a}}

\newcommand{\fS}{{\mathfrak S}}
\newcommand{\rH}{{\rm H}}

\newcommand{\cout}[1]{}

\newcommand{\sgn}[1]{\,{\rm sign}(#1)}

\newcommand{\tr}{{\rm tr}}
\newcommand{\Tr}{{\rm Tr}}

\newcommand{\hpdo}{{$h\Psi$DO}}
\newcommand{\pdo}{{$\Psi$DO}}
\newcommand{\R}{{\rm R}}

 \renewcommand{\arraystretch}{1.5}

\title{Topological invariants for interface modes}

\author{Guillaume Bal \thanks{University of Chicago; {\tt guillaumebal@uchicago.edu}}}

\begin{document}
 
\maketitle

\tableofcontents

\begin{abstract}
We consider topologically non-trivial interface Hamiltonians, which find several applications in materials science and geophysical fluid flows. The non-trivial topology manifests itself in the existence of topologically protected, asymmetric edge states at the interface between two two-dimensional half spaces in different topological phases. It is characterized by a quantized interface conductivity. The objective of this paper is to compute such a conductivity and show its stability under perturbations.  We present two methods. 

The first one computes the conductivity using the winding number of branches of absolutely continuous spectrum of the interface Hamiltonian. This calculation is independent of any bulk properties but requires a sufficient understanding of the spectral decomposition of the Hamiltonian. In the fluid flow setting, it also applies in cases where the so-called bulk-interface correspondence fails.

The second method establishes a bulk-interface correspondence between the interface conductivity and a so-called bulk-difference invariant. We introduce the bulk-difference invariants characterizing pairs of half spaces. We then relate the interface conductivity to the bulk-difference invariant by means of a Fedosov-H\"ormander formula, which computes the index of a related Fredholm operator. 

The two methods are used to compute invariants for representative $2\times2$ and $3\times3$ systems of equations that appear in the applications.
\end{abstract}
 

\renewcommand{\thefootnote}{\fnsymbol{footnote}}
\renewcommand{\thefootnote}{\arabic{footnote}}

\renewcommand{\arraystretch}{1.1}



\noindent{\bf Keywords}:  Interface Hamiltonian, Fredholm operator, index theory, bulk-difference invariant, bulk-interface correspondence, semiclassical calculus, topological insulators, topological waves.



\section{Introduction}
\label{sec:intro}

Topological invariants offer a useful description of phenomena that appear to be immune to large classes of perturbations. The macroscopic transport properties of topological insulators are for instance known to be dictated to a large extent by a topological invariant characterizing their phase \cite{bernevig2013topological,Fruchart2013779,prodan2016bulk}. More recently, the behavior of equatorial waves and that of other wave phenomena was also shown to afford a topological description \cite{delplace2017topological,shankar2017topological,souslov2019topological,tauber2019bulk}. Its main practical consequence is the stability of some macroscopic transport properties in the presence of defects. This paper addresses the computation of these invariants and their relations to quantized physical observables for continuous two-dimensional models that appear in the analysis of topological insulators and topological waves.

We consider Hamiltonians generically denoted by $H[\mu]$ with $\mu$ an order parameter taking values in $\Rm$ (or more general manifolds).  A typical Hamiltonian that appears in the analysis of topological insulators with $\mu=m$ a mass term is the $2\times2$ Dirac system
\begin{equation}\label{eq:s22}
  H[m(y)] =D\cdot\sigma + m(y) \sigma_3, \qquad H[m] = k\cdot\sigma + m\sigma_3 = \left(\begin{matrix} m & \xi-i\zeta \\ \xi+i\zeta & -m\end{matrix} \right),
\end{equation}
in the physical and Fourier variables (when $m$ is constant), respectively. Here, $\bx=(x,y)$ are the spatial variables, $k=(\xi,\zeta)$ the dual Fourier variables, $D=(D_x,D_y)$ with $D_x=\frac1i\partial_x$ and $D_y=\frac1i\partial_y$ while $\sigma=(\sigma_1,\sigma_2)$ and $\sigma_{1,2,3}$ are the standard Pauli matrices. 
The above system as well as more general (and more physical) block-diagonal direct sums of such elementary blocks appear naturally (and generically) in the analysis of topological insulators \cite{B-BulkInterface-2018,B-EdgeStates-2018,drouot2019bulk,drouot2018defect,Fefferman2016}.

For equatorial waves, $\mu=f$ is the Coriolis force (with as a possible second order parameter an odd viscosity term $\epsilon$ in which case $f$ below is replaced by $f+\varepsilon\Delta$ \cite{delplace2017topological,souslov2019topological,tauber2019bulk}). The typical $3\times3$ system we consider is
\begin{equation}\label{eq:s33}
  H[f(y)] = \big(D_x,D_y,-f(y)\big) \cdot \Gamma,\qquad H[f] = (\xi,\zeta,-f) \cdot \Gamma = \left(\begin{matrix} 0&\xi&\zeta\\ \xi&0&if \\ \zeta&-if&0\end{matrix} \right),
\end{equation}
where $\Gamma=(\gamma_1,\gamma_4,\gamma_7)$ involves some of the Gell-Mann matrices used in the representation of the Lie algebra $\mathfrak{su}(3)$; see \cite{delplace2017topological,souslov2019topological} for derivations of such a fluid model and its applications.\footnote{The physical system $H[f(y)]$ is the same as in  \cite{delplace2017topological,souslov2019topological,tauber2019bulk} while its Fourier domain representation $H[f]$ is not. Our convention for plane wave representations is $e^{i(k\cdot x-Et)}$ (rather than $e^{i(Et-k\cdot x)}$) so that $k=(\xi,\zeta)$ is the Fourier multiplier of $D=\frac1i\nabla$, a non-negotiable convention in semiclassical analysis. This choice affects the numerical values of our topological invariants. The dispersion relations $E=E(k)$ and group velocities $\nabla_kE(k)$ are, however, the same in both representations.} The above system applies to a vector field $(\eta,u,v)$ with $\eta$ the (variation of the) thickness of the two dimensional fluid and $(u,v)$ its velocity components.


\medskip

The main objectives of the paper are to introduce an interface invariant describing the asymmetric transport along the interface and to devise means to compute it. 

Physically interesting phenomena appear when insulating domains with different bulk topologies are brought next to each other. Consider a material modeled by a Hamiltonian $H[\mu]$. It is said to be insulating in the energy range $[E_1,E_2]$ when the latter interval and the spectrum of $H[\mu]$ do not intersect. A bulk topology may then be assigned to $H[\mu]$ \cite{bernevig2013topological,Fruchart2013779,prodan2016bulk}; we do so in section \ref{sec:bulk} below. Consider two materials modeled by $H[\mu_\pm]$ in the half spaces $\pm y>1$ with bulk topologies characterized by integers $I_\pm=I(\mu_\pm)$ (with a (possibly) smooth transition between the two). The interface will be modeled more generally by a varying order parameter $\mu=\mu(y)$ and corresponding interface Hamiltonian $H[\mu(y)]$. At the interface $y\sim0$, the bulk invariant is {\em not} defined, the energy gap closes, and transport along the interface is allowed. Quite surprisingly, transport is typically asymmetric, with more modes moving in one direction than in the other one. 

The domain walls $\mu(y)$ and the construction of spatial truncations and of density of states within the bulk band gap $[E_1,E_2]$ are all conveniently represented as {\em switch functions} $f\in{\mathfrak S}[f_-,f_+,x_-,x_+]$ which we define as the family of bounded functions from $\Rm$ to $\Rm$ such that
\begin{equation}\label{eq:switch}
    f(x)=f_- \quad \forall x\leq x_- \qquad \mbox{ and } \qquad  f(x)=f_+ \quad \forall x\geq x_+.
\end{equation}
We also denote by ${\mathfrak S}[f_-,f_+]$ the union of the above switch functions over all $x_-\leq x_+$. A {\em smooth switch function} is a switch function in $C^\infty(\Rm)$. A domain wall $\mu(y)$ is then typically an element in $\fS[\mu_-,\mu_+]$. We also consider domain walls of the form $\mu(y)=\lambda y$ for $\Rm\ni \lambda\not=0$.

This topologically protected asymmetric transport is characterized by what we call an interface conductivity, in analogy to the linear response to radiation in condensed matter physics \cite{bernevig2013topological,Fruchart2013779,prodan2016bulk}. This conductivity is defined for a Hamiltonian $H=H[\mu(y)]$ as 
\begin{equation}\label{eq:sigmaI}
  \sigma_I = \sigma_I[H] := {\rm Tr}\ i[H,P]\varphi'(H).
\end{equation}
Here, $P=P(x)$ is a smooth switch function from $\Rm$ to $[0,1]$ in $\fS[0,1,x_0-\beta,x_0+\beta]$ for some $\beta>0$. The conductivity $\sigma_I$ is in fact independent of $x_0$ and $\alpha$. Above, $\varphi(E)$ is constructed as a smooth switch function in $\fS[0,1,E_1,E_2]$ such that $\varphi'(E)\geq0$ to afford an interpretation as a density of states while $\varphi'(H)$ is then constructed by spectral calculus. The interval $[E_1,E_2]$ is chosen to lie in a bulk band gaps of $H[\mu_\pm]$. 

The interpretation of $\sigma_I$ is as follows. $\varphi'(H)$ may be seen as a (nonnegative) density of states (integrating to $1$). The range of energies in the system is restricted to this interval to avoid leakage into the bulk. Then $P$ may be seen as an observable with $\langle\psi|P|\psi\rangle$ quantifying an amount of signals to the right of $x_0$.  Its evolution $\langle\psi|\dot P|\psi\rangle:=\frac{d}{dt}\aver{e^{-it H}\psi|P|e^{-it H}\psi}=\frac{d}{dt}\aver{\psi|e^{it H}Pe^{-it H}|\psi}$ in Hamiltonian dynamics results in the Heisenberg formalism as $\dot P=i[H,P]$ with $[A,B]=AB-BA$ the usual commutator. The above (operator) trace then has the clear physical interpretation (if defined) as a rate of signal moving from the left of $x_0$ to the right of $x_0$. This  conductivity has appeared in different contexts in many works on topological protection of edge states \cite{B-BulkInterface-2018,drouot2019microlocal,elbau2002equality,elgart2005equality,prodan2016bulk}.

It turns out that $\sigma_I$ is often quantized ($2\pi\sigma_I\in\Zm$) with an invariant that is immune to perturbations in the definition of $H$.  
Moreover, the asymmetry is quantified by a general principle, the bulk-interface correspondence, which states that the excess of modes going in one direction versus the other is given by $I(\mu_+)-I(\mu_-)$.

While appealing, this correspondence does not (quite) always apply in the case of equatorial waves \cite{delplace2017topological}. In particular, the number of edge states depends on the structure of the Coriolis force term $f(y)$ with a number of asymmetric interface modes equal to $2$ when $f(y)=\lambda y$ and equal to only one when $f(y)=f_0\sgn{y}$. The main heuristic reason for such difficulties comes from the fact that the topological change occurs when $f(y)$ changes signs. For such a value, the gap closing of the Hamiltonian $H[f]$ coincides with a sheet of infinitely degenerate spectrum (modes in geostrophic balance corresponding to an infinitely degenerate $0\in {\rm Sp}(H[f=0])$), whose behavior is influenced by the profile of the variation $f(y)$.

\medskip

This paper introduces two methodologies to compute the interface conductivity $\sigma_I$. The first one, presented in section \ref{sec:ii}, extends the work in \cite{B-BulkInterface-2018} on the system \eqref{eq:s22} to more general Hamiltonians including the one in \eqref{eq:s33}. The method is based on mapping $H[\mu(y)]$ to an appropriate Fredholm operators as in \cite{B-BulkInterface-2018,prodan2016bulk}. The index of the operator is then related to a sequence of winding numbers associated to the Hamiltonian's continuous spectrum. Such invariants are defined independently of any bulk invariant. Moreover, the method allows us to compute $\sigma_I$ for $H$ defined in \eqref{eq:s33} both when  $f(y)=\lambda y$ and when $f(y)=f_0\sgn{y}$. The main drawback of the method  is that it requires explicit knowledge of the structure of the branches of absolutely continuous spectrum. Once these invariants are computed, one shows as in  \cite{B-BulkInterface-2018}  that they are stable with respect to appropriate classes of perturbations  $H_V:=H[\mu(y)]+V$.

\medskip

The second methodology is based on a new bulk-interface correspondence relating the above interface conductivity to the symbol of an appropriate Fredholm operator whose index is computed directly from its symbol in a Fedosov-H\"ormander formula. This correspondence may be interpreted as a topological charge conservation as described in \cite{gurarie2011single,volovik2009universe}. A reinterpretation of the formula then relates the conductivity to several equivalent expressions of a quantity we call a bulk-difference invariant, which generalizes the difference of bulk invariants we mentioned above. 

The bulk-difference invariant is introduced in section \ref{sec:bulk}. We present in section \ref{sec:BI} a semiclassical analysis of the interface conductivity and show how to relate it to the bulk-difference invariant and to an index theory similar to that of Fedosov-H\"ormander \cite[Chapter 19]{H-III-SP-94}. The correspondence shares some similarities with the bulk-interface correspondence based on spectral flows and spectral asymmetry in the physics literature \cite{fukui2012bulk,gurarie2011single,volovik2009universe}; see Appendix \ref{sec:bic}.  Our constructions also show that the interface conductivity is stable with against perturbations $V$. 

As an application to the system \eqref{eq:s33}, we obtain that the bulk-interface correspondence applies provided that $f(y)$ varies sufficiently slowly. This is consistent with the explicit computations obtained when $f(y)=\lambda y$ and tends to indicate that the bulk-interface correspondence does not apply when $f(y)=f_0\sgn{y}$ because the latter is too singular.

%
%
%
%
%
\section{Spectral computation of interface invariants}
\label{sec:ii}
%
%

This section computes the interface conductivity $\sigma_I$ defined in \eqref{eq:sigmaI} for Hamiltonians $H=H[\mu(y)]$ in terms of the winding number of their branches of absolutely continuous spectrum. This extends work in \cite{B-BulkInterface-2018} on the system \eqref{eq:s22}. 

The order parameter $\mu=\mu(y)\in \fS[\mu_-,\mu_+]$ generates a domain wall in the topologically non-trivial case $\mu_+\mu_-<0$. The non-vanishing values of $\mu$ act as insulators while the interface in the vicinity of $y=0$ where $\mu$ changes signs becomes conducting. This is the physical origin for the topologically protected asymmetric transport.

The results in \cite{B-BulkInterface-2018} show that we can relate $\sigma_I$ to the index of a Fredholm operator and that both remain quantized when $H$ is replaced by $H+V$ for $V$ in an appropriate class of perturbations. 
The computation of $\sigma_I$ based on  branches of absolutely continuous spectrum is independent of any definition of bulk invariants and hence of any bulk-interface correspondence. However, it is useful in practice when the branches can be analyzed in sufficient detail. The analysis carried out in \cite{B-BulkInterface-2018} for the system \eqref{eq:s22} applies to general mass terms $m(y)$. However, the algebra for the $3\times3$ system in \eqref{eq:s33} is more complex and carried out only in specific situations such as when $f(y)=\lambda y$ and $f(y) = f_0 \sgn{y}$. The details of these calculations are presented in Appendix \ref{sec:appli}.

\medskip

We consider a Hamiltonian $H=H[\mu(y)]$ acting on vectors $(\psi_k(x,y))_{1\leq k\leq N}$ with coefficients $\mu=\mu(y)$ independent of $x$ so that the following spectral decomposition holds:
\begin{equation}\label{eq:Hxi}
  H = {\mathcal F}^{-1}_{\xi\to x} \dint_{\Rm}^{\oplus} H[\xi] d\xi {\mathcal F}_{x\to\xi}.
\end{equation}
Here ${\mathcal F}$ is the one dimensional Fourier transform in the $x$-variable.
We assume that $H$ is an unbounded self-adjoint operator from its domain of definition ${\mathcal D}(H)$ to $L^2(\Rm^2)\otimes\Cm^N$. The domain of definition is defined as $(H\pm i)^{-1} (L^2(\Rm^2)\otimes\Cm^N)$, which is a subspace of $L^2(\Rm^2)\otimes\Cm^N$ independent of $\pm$ \cite[Theorem 1.2.7]{davies_1995}. We verify in Appendix \ref{sec:appli} that this property holds \eqref{eq:s22} and \eqref{eq:s33}. Since $H$ is self-adjoint, we have access to the corresponding spectral calculus; see \cite{davies_1995} and Appendix \ref{sec:hpdo}.

Let $(E_1,E_2)$ be a fixed interval. We assume:
\\[2mm]
 [H-AC]: (i) for each $\xi\in\Rm$, $H[\xi]$ restricted to the interval $[E_1,E_2]$ has finite rank and vanishes for $\xi$ outside of a compact set; and (ii)  the existence of smooth (in $\xi$) corresponding eigenvalues $E_j(\xi)$ and rank-one projectors $\Pi_j(\xi)$ (with Schwartz kernel $\psi_j(y,\xi)\psi_j^*(y',\xi)$ for normalized eigenvectors $\int_{\Rm}|\psi(y,\xi)|^2dy=1$) parametrizing a finite number $J$ of branches of absolutely continuous spectrum.  
\\[2mm]
The above branches $E_j(\xi)$ are so far defined on an implicit compact domain $\Xi_j\ni\xi$. They are extended by continuity  to $\xi\in\Rm$ as the unique continuous branches still called $E_j(\xi)$ such that $E_j(\xi)\in\{E_1,E_2\}$ for $\xi\not\in \Xi_j$.

We now construct a unitary operator from the restriction of $H$ to the interval $[E_1,E_2]$. Let $\varphi(E)$ be a smooth switch function in $\fS[0,1,E_1,E_2]$. We then define the unitary operator
\[
   U(H) = e^{i2\pi \varphi(H)},\qquad W(H)=U(H)-I
\]  
by spectral calculus \cite{dimassi1999spectral}.  Note that $W(E)$ is compactly supported in $(E_1,E_2)$.  By assumption, we have the decomposition
\[
  W(H) = {\mathcal F}^{-1} \dint_{(E_1,E_2)}^{\oplus} \dsum_{j=1}^J W(E_j(\xi)) \Pi_j(\xi) d\xi\  {\mathcal F} 
 =  {\mathcal F}^{-1} \dint_{\Rm}^{\oplus} \dsum_{j=1}^J W(E_j(\xi)) \Pi_j(\xi) d\xi \ {\mathcal F} .
\]

Let finally $P$ be either the function $P$ defined in the introduction or a spatial projector onto $x\geq x_0$ for $x_0\in\Rm$, i.e., point-wise multiplication by $\rH(x-x_0)$, the Heaviside function. Then, as an extension of the results \cite{B-BulkInterface-2018}, we have 
\begin{theorem}\label{thm:intunpert}
  Let $P$ and $U(H)$ be defined as above for $H$ satisfying the hypothesis [H-AC](i)-(ii). Then $PU(H)P$ is a Fredholm operator on the range of $P$. Moreover, 
\begin{equation}\label{eq:intunpert}
  I[H]:={\rm Index}(PUP) = -{\rm Tr} [P,U]U^* =  \dsum_{j=1}^J {\mathcal W}_1(e^{i 2\pi \varphi\circ E_j})
\end{equation}
with ${\mathcal W}_1(f)$ the winding number of a unimodular complex function $f$  with compactly supported gradient, given explicitly by
\begin{equation}\label{eq:W1}
  {\mathcal W}_1(f) = \dfrac{1}{2\pi i} \dint_{\Rm} \partial_\xi f f^*(\xi) d\xi.
\end{equation}
\end{theorem}
The proof of this theorem is given in Appendix \ref{sec:appindex}.
Applying the above theorem yields 
\[
  I[H]=\Tr [U,P]U^* = N_+ - N_-,
\]
where $N_+$ is the number of branches $E_j$ such that $E_j(\xi)\in\fS[E_1,E_2]$ while $N_-$ is the number of branches  $E_j$ such that $E_j(\xi)\in\fS[E_2,E_1]$. Branches such that  $E_j(\xi)$ belongs to $\fS[E_1,E_1]$ or $\fS[E_2,E_2]$
 do not contribute to $I[H]$.

\medskip
 \noindent
{\bf Stability and interface conductivity.} 
One of the main reasons to develop the above apparatus borrowed from non-commutative geometry (see \cite{B-BulkInterface-2018,prodan2016bulk} and their references)  is that it allows us to include spatial perturbations in the Hamiltonians. Let $H$ be as above and $H_V=H+V$ a perturbed Hamiltonian.  

We recall results from \cite{B-BulkInterface-2018} for completeness.  Stability of the interface conductivity $\sigma_I$ is guaranteed when $V$ is relatively compact with respect to $H$, which means that $V(H+i)^{-1}$ is a compact operator in the $L^2$ sense. For the $2\times2$ problem \eqref{eq:s22}, any $V=V(x,y)$ an operator of point-wise multiplication by a bounded function that decays to $0$ at infinity satisfies such hypotheses since the symbol of $(H+i)^{-1}$ decays to $0$ as $(\xi,\zeta)\to\infty$; this is essentially an application of Rellich's compactness criterion \cite{simon2015operator}.
 In other words, the topology is fixed by our assumptions `at infinity' but $V$ can be arbitrarily large (bounded) so long as it decays at infinity. For the $3\times3$ system \eqref{eq:s33}, the symbol of $(H+i)^{-1}$ does not converge to $0$ as $(\xi,\zeta)\to\infty$ and it is more difficult to characterize operators $V$ for which the above criterion holds. 

Let us define $U[H_V]= e^{i2\pi \varphi(H_V)}$ spectrally. Then we have \cite{B-BulkInterface-2018}
\begin{theorem}\label{thm:intpert}
$PU[H_V]P$ on the range of $P$ is a Fredholm operator and 
\begin{equation}\label{eq:indexpert}
  I[H_V] :={\rm Index}(PU[H_V]P) =  I[H].
\end{equation}
\end{theorem}
The interface index is thus independent of any relatively compact perturbation.

Still following \cite{B-BulkInterface-2018}, we now relate the above index to the interface conductivity in \eqref{eq:sigmaI}. The latter is also stable with respect to perturbations $V$ although it is less stable than the above index $I[H_V]$. Let us now impose  that $V$ is of the form $V=V_1V_2$ with
\begin{equation}\label{eq:Vtrace}
 \|V_j\|\leq C \ \mbox{ and } \ \|(z-H)^{-1}V_j\|_{HS} \leq C |\Im z|^{-1}\quad j=1,2,
\end{equation}
where $HS$ is the Hilbert-Schmidt norm. It is then proved in \cite{B-BulkInterface-2018} that $[H_V,P] \varphi'(H_V)$ is indeed a trace-class operator and we have:
\begin{theorem}\label{thm:sigmaI}
  Let $H$ be as above and $V$ a perturbation satisfying \eqref{eq:Vtrace}. Then  the operator $i[H_V,P] \varphi'(H_V)$ is trace-class and 
\[
  \sigma_I := i {\rm Tr} [H_V,P] \varphi'(H_V) = \frac1{2\pi} I[H_V].
\]
\end{theorem}

This provides a means to compute the topologically protected conductivity $\sigma_I$ by spectral analysis of the unperturbed Hamiltonian $H[\mu(y)]$ using Theorem \ref{thm:intunpert} above. 

That the interface conductivity is quantized is consistent with the numerical simulations presented in, e.g.,  \cite{delplace2017topological,souslov2019topological}.  
The calculations of interface conductivities for the systems \eqref{eq:s22} and \eqref{eq:s33} are presented in Appendix \ref{sec:appli}.  A summary for system \eqref{eq:s22} is that $2\pi\sigma_I=-\sgn{m_+}$ when $m(y)$ is a domain wall between $m_-$ and $m_+$ when $m_+m_-<0$ (with $\sigma_I=0$ when $m_+m_->0$) while $2\pi\sigma_I=-\sgn{\lambda}$ when $m(y)=\lambda y$. For system \eqref{eq:s33}, we find that $2\pi\sigma_I$ is given by $\sgn{f_0}$ when $f(y)=f_0\sgn{y}$  while it is given by $2\sgn{\lambda}$ when $f(y)=\lambda y$.

This discrepancy is reminiscent of correction terms in Levinson's theorem caused by the presence of resonances for the Hamiltonian at energy $0$ \cite{graf2020topology,kellendonk2012wave}.

\section{Bulk-difference invariant}
\label{sec:bulk}

Bulk invariants are defined when the coefficients $\mu$ such as $m$ or $f$ above are constant. 
They are of interest in their own right and may be extended to operators with spatially varying coefficients as in \cite{B-BulkInterface-2018} and in a variety of other settings \cite{prodan2016bulk}. They are also instrumental in bulk-boundary correspondences \cite{bourne2017k,bourne2018chern,drouot2019bulk,elbau2002equality,Graf2013,graf2018bulk,hatsugai1993chern, kellendonk2004quantization,ludewig2020cobordism,prodan2016bulk}, which relate them to the interface conductivity $\sigma_I$.

With $\mF=\mF_{(x,y)\to k=(\xi,\zeta)}$ the two-dimensional Fourier transform, we may then write the spectral decomposition
\[
  H = {\mathcal F}^{-1} \dint_{\Rm^2}^{\oplus} H(k) dk {\mathcal F}
\]
where $H(k)$ is $M_n(\Cm)$-valued (square matrices of dimension $n$) with $n=2$ and $n=3$ in \eqref{eq:s22} and \eqref{eq:s33}, respectively.
We then have the spectral decomposition
\begin{equation}\label{eq:diagH}
   H(k) = \dsum_{i=1}^n h_i(k) \Pi_i(k)
\end{equation}
with eigenvalues $h_i(k)\in\Rm$ for $k\in\Rm^2$ ordered with increasing values and rank-one orthogonal projectors $\Pi_i(k)= \psi_i(k) \psi_i^*(k)$ for $1\leq i\leq n$. We assume that $h_i(k)$ and $\Pi_i(k)$ are sufficiently smooth (e.g. $C^2$) in the variable $k$. This holds for $H(k)$ smooth and in the presence of (local) spectral gaps separating the energies $i\mapsto h_i(k)$.  

For the $2\times2$ problem in \eqref{eq:s22}, we even have global spectral gaps since  $h_{1,2}(k)=\mp\sqrt{|k|^2+m^2}$ with one spectral gap given by $(-|m|,|m|)$ while for the $3\times3$ problem in \eqref{eq:s33}, $h_{1,2,3}(k) = (-\sqrt{|k|^2+f^2},0,\sqrt{|k|^2+f^2})$ with two global spectral gaps given by $(-|f|,0)$ and $(0,|f|)$.   The (smooth) rank-one projectors are also well known; see \eqref{eq:eigen3} below for the $3\times3$ system as well as \cite{delplace2017topological,Fruchart2013779,souslov2019topological,tauber2019bulk}. In both cases, $\Pi_1\equiv \Pi_-$ corresponds to the projection of the Hamiltonian onto its (strictly) negative spectrum. 

\medskip

Let $\Pi(k)$ be any sum $\Pi(k)=\sum_{i\in I}\Pi_i(k)$ for $I$ a subset of $\{1,\ldots,n\}$ independent of $k$. 
Then $k\mapsto \Pi(k)$ defines a vector bundle over the base manifold $\Rm^2$ \cite{Fruchart2013779,nakahara2003geometry}. When $\Rm^2$ is replaced by a compact manifold $M$, then vector bundles over $M$ admit topological classifications based on their Chern classes. For the above projectors, living in spaces of matrices over $M$, such classes integrated over the base manifold give rise to integer-valued objects called Chern numbers:
\[
  \tilde c[\Pi]= \frac{i}{2\pi} \dint_M {\rm tr} \Pi d\Pi \wedge d\Pi = \frac{i}{2\pi} \dint_M {\rm tr} \Pi [\partial_1\Pi,\partial_2 \Pi] d^2k  \  \in \Zm.
\]
Here, $[A,B]=AB-BA$ and tr refers to the standard matrix trace.

When $M$ is not compact, for instance $\Rm^2$, the domain of interest in this paper, the above integrals can often still be evaluated but are no longer guaranteed to be integer-valued \cite{B-BulkInterface-2018,delplace2017topological,Fruchart2013779,souslov2019topological,tauber2019bulk}. For instance, for the $2\times2$ problem, $\tilde c[\Pi_1]=\frac12\sgn{m}$. A possible solution to this issue is to regularize the Hamiltonian and its projectors in such a way that they take a unique value as $|k|\to\infty$. This allows one to compactify the plane around the unit sphere mapping $\infty$ to the south pole, say, and still obtain a continuous family of projectors on the sphere. The above integral, which is manifestly invariant by change of variables from its $PdP\wedge dP$ form, may then be computed for $M=\Sm^2$ and shown to be integral. A typical regularization consists in replacing $m$ by $m-\eta|k|^\alpha$ for $\alpha>1$ and $\eta\not=0$. We may then show that $\tilde c[\Pi_1]=\frac12(\sgn{m}+\sgn{\eta})\in\Zm$ \cite{B-BulkInterface-2018}.

A similar regularization, based on replacing $f$ by $f-\eta|k|^2$ in the $3\times3$ model \eqref{eq:s33}, also allows one to define a topological invariant on the sphere $\Sm^2$. One then finds that $\tilde c[\Pi_3]=-\tilde c[\Pi_1]= (\sgn{f}+\sgn{\eta})\in\Zm$ \cite{souslov2019topological,tauber2019bulk}. 

The bulk invariants of such regularized operators therefore depend on the sign of the regularization. The main advantage of this regularization is that it still applies in the setting where $H$ is no longer translationally invariant \cite{B-BulkInterface-2018}. We now propose a different method to bypass this somewhat artificial regularization term and define topological invariants in the specific situation of interest here, namely the analysis of interface Hamiltonians, with one half space essentially corresponding to a bulk invariant coming from $m_+$ or $f_+$ and the other half space corresponding to a (possible) sign change of the mass terms $m_-$ or $f_-$. More generally, let us assume the existence of two Hamiltonians
\begin{equation}\label{eq:Hpm}
   H^\pm(k) = \dsum_{i=1}^n h^\pm_i(k) \Pi^\pm_i(k)
\end{equation}
with smooth projectors $k\to\Pi_i^\pm(k)$. The structure of the energies $h_i^\pm(k)$ is irrelevant beyond the existence of well-defined gaps and each $h_i(k)$ may be continuously modified (homotopically transformed) to a single value $h_i$ (flat band) that depends on neither $k$ nor $\pm$. In the applications of interest here, the two Hamiltonians are identical except for the value of their mass terms $\mu_\pm$. 

\medskip\noindent
{\bf Construction of the bulk-difference invariants.}
Let $\Pi^\pm(k)= \sum_{i\in I} \Pi^\pm_i(k)$ for $I\subset\{1,\ldots,n\}$ independent of $k\in\Rm^2$ be two smooth families of projectors.
Defining $k=|k|\theta$, we {\em assume} the continuous matching (gluing condition) of the projectors in all directions at infinity:
\begin{equation}\label{eq:sphere}
  \lim_{|k|\to\infty} \Pi^+(|k|\theta) =  \lim_{|k|\to\infty} \Pi^-(|k|\theta) \qquad \mbox{ for all } \theta\in \Sm^1.
\end{equation}
We assume that these limits exist and are continuous in $\theta$. In our applications, the projectors at $\infty$ do not depend on the mass terms $m$ or $f$ and thus satisfy the above hypothesis.

We then {\em define} a new projector $\Pi(k)$ for $k$ an element in the union of two planes $P_\pm\simeq\Rm^2$ that are wrapped around the unit sphere $\Sm^2\simeq(P_+\sqcup P_-)\slash \sim $ so that the circles at infinity are glued (identified by $\sim$) along the sphere's equator. For $k\in P_\pm$, we define $\Pi(k)=\Pi^\pm(k)$. For a point $\phi$ on the sphere, a form of stereographic projection $\pi$ maps $\phi$ in the upper half sphere to $k\in P_+$ and $\phi$ in the lower half sphere to $k\in P_-$. More precisely, with $\phi\in \Sm^2$ parametrized by $(x,y,z)$, we have
\[
   (x,y) = \frac{k}{\sqrt{1+|k|^2}},\quad z=\frac{\pm1}{\sqrt{1+|k|^2}}, \qquad k\in P_\pm,
\] 
with $\pi$ the inverse map, i.e., $k=\pi(\phi)$. We then define $\pi^*\Pi(\phi)=\Pi(\pi(\phi))$ the pull back by $\pi$ (still called $\Pi(\phi)\equiv\pi^*\Pi(\phi)$ to simplify notation) a projector that is now continuous on $\Sm^2$ thanks to the continuity assumption \eqref{eq:sphere}. We may therefore define the Chern numbers as integrals over the sphere, a compact cycle, which written on the sphere and then pushed by $\pi$ to the planar variables, are given by
\begin{equation}\label{eq:cbd}
  c[\Pi] =  \dfrac{i}{2\pi} \dint_{\Sm^2} {\rm tr} \Pi d\Pi \wedge d\Pi =  \dfrac{i}{2\pi} \dint_{\Rm^2} {\rm tr} \Big( \Pi^-[\partial_1\Pi^-,\partial_2\Pi^-] - \Pi^+[\partial_1\Pi^+,\partial_2\Pi^+]\Big)  dk,
\end{equation}
where the $-$ sign above is necessary to ensure that $\Sm^2$ has a given orientation, here inherited from that of the lower plane $P_-$ and opposite that of the upper plane $P_+$. This ensures that $\Sm^2$ also inherits its orientation from the $dx\wedge dy\wedge dz>0$ positive orientation of $\Rm^3$ it is embedded in.
\begin{definition}[Bulk-difference invariant]\label{def:bulkdifference}
 Let $H^\pm(k)$ be decomposed as in \eqref{eq:Hpm}  and let $\Pi^\pm(k)= \sum_{i\in I} \Pi^\pm_i(k)$ for $I\subset\{1,\ldots,n\}$ satisfying the  gluing condition \eqref{eq:sphere}. Then the above integrals $c[\Pi]$ are well defined integers we call the   {\em bulk-difference} invariants. We define $c_i=c[\Pi_i]$ when $I=\{i\}$.
\end{definition}
These invariants are by construction immune to any (continuous) perturbation of $H(k)$ that maintains the spectral gaps  and the gluing assumptions as $|k|\to\infty$ and satisfy the general  additivity property of Chern numbers  \cite{avron1983homotopy} resulting from the additivity property of Chern classes \cite{nakahara2003geometry}:
\[
  c[\Pi_i+\Pi_{i+1}] = c[\Pi_i] + c[\Pi_{i+1}].
\]
When the two Hamiltonians satisfy $H_+=H_-$, then the above integral vanishes. However, for the system $H_\pm=k\cdot\sigma  +m_\pm\sigma_3$, which satisfies all the above assumptions, we find, following calculations as in e.g.  \cite{B-BulkInterface-2018} that
\begin{equation}\label{eq:chern2}
 c_- :=c[\Pi_1] = -c_+ :=- c[\Pi_2] = \frac12 \sgn{m_-} - \frac12 \sgn{m_+} \in\Zm
\end{equation}
which is a bona fide invariant even in the absence of regularization ($\eta=0$ above). For the system \eqref{eq:s33}, which also satisfies the above assumptions, we find that 
\begin{equation}\label{eq:chern3}
  c_+ :=c[\Pi_3]= c_0+c_+ := c[\Pi_3+\Pi_2]= - c_- :=- c[\Pi_1]= \sgn{f_-} - \sgn{f_+}  \in \Zm.
\end{equation}
The calculations are given in some detail in Appendix \ref{sec:appeq}; see \eqref{eq:eigen3} and the following computations. For the above problem, only $c_+$ needs to be computed since $c_-=-c_+$ by symmetry and $c_++c_0+c_-=0$. 
\medskip

Explicit calculations of Chern numbers may be obtained in several ways. The first one is to compute the Berry curvature $i{\rm tr}\ \Pi_i d\Pi_i \wedge d\Pi_i$ directly (from $\psi_i(k)$) and integrate it over the plane(s) $\Rm^2$, as done for instance in \cite{B-BulkInterface-2018,delplace2017topological,souslov2019topological}. The second method consists in directly looking at the line bundle generated by $\psi_i(k)$ and how charts covering the sphere need to be glued by an appropriate transition of connections to respect the twists of the eigenvectors \cite{bernevig2013topological,Fruchart2013779,tauber2019bulk}. The third method, which is the most versatile when it applies, recasts the Chern number as the degree of a map $k\to {\rm h}(k)=(\xi,\zeta,m)\in\Rm^3\backslash\{0\}$ \cite{B-BulkInterface-2018,Fruchart2013779,graf2020topology}.  In particular, it applies to system \eqref{eq:s33} seen as a spin-$1$ representation and provides the computations in \eqref{eq:chern3} \cite{graf2020topology}.

\medskip\noindent{\bf Green's function invariant.}
We now consider a different form of the bulk-difference  invariant based on the notion of resolvent or Green's function \cite{gurarie2011single,volovik2009universe} and similar to calculations based on the Kubo formula \cite{bernevig2013topological}. It is given by
\[
   G=G_\alpha(\omega,k) = (z-H(k))^{-1} = \dsum_{i=1}^n (z-h_i(k))^{-1} \Pi_i(k)
\]
for $k$ in a domain $K$ and $z=\alpha+i\omega$ with $\alpha$ a fixed real number in a {\em global} spectral gap, i.e., $\alpha\not=h_i(k)$ for all $1\leq i\leq n$ and $k\in K$. Therefore, $G$ and $G^{-1}$ are well-defined with obviously $G^{-1}(k)=z-H(k)$. 

Here, we consider $K=\Sm^2\simeq(\Rm^2\sqcup\Rm^2)\slash \sim$, the sphere introduced earlier. We thus assume that for each $1\leq i\leq n$, $\Pi_i^\pm(k)$ satisfy the gluing conditions \eqref{eq:sphere}. In addition, we assume that $(z-h_i^\pm(k))^{-1}$ also satisfy the gluing conditions  \eqref{eq:sphere}. Thus, the resulting $(z-h_i)^{-1} \Pi_i$ is continuous on $\Sm^2\simeq(\Rm^2\sqcup\Rm^2)\slash \sim$. We identify $\omega\in\Rm$ with its one point compactification  $\Rm\slash\sim\simeq\Sm^1$ where the points at infinity $\omega\to\pm\infty$ are identified. We observe that $(z-h_i)^{-1}$ and hence $G$ is also continuous on  $\Rm\slash\sim\simeq\Sm^1$. We have thus constructed a continuous function from $M=\Sm^1\times\Sm^2 \simeq (\Rm\slash \sim)\times(\Rm^2\sqcup\Rm^2)\slash \sim$ to $GL_n(\Cm)$ the space of $n-$dimensional invertible matrices. 

Associated to the Green's function is the three dimensional (generalized) winding number
\begin{equation}\label{eq:Walpha}
 W_\alpha = \dfrac{1}{24\pi^2}  \dint_M {\rm tr}  (dG^{-1} G)^{\wedge 3}  =  \dfrac{1}{8\pi^2} \dint_M  {\rm tr}\partial_\omega G^{-1}G [ \partial_1 G^{-1}G, \partial_2 G^{-1}G] d\omega d^2k,
\end{equation}
where by a slight abuse of notation, the above right-hand side should be interpreted as an integration over the two half planes $P_\pm$ as in \eqref{eq:cbd}. Since the groups \cite{nakahara2003geometry} $\pi_3(\Sm^1\times \Sm^2)=\pi_3(\Sm^1)\oplus \pi_3(\Sm^2) =\Zm$, the (generalized) three-dimensional winding number of $G$ is a well-defined integer. 

The above winding number is related to the Chern numbers we previously defined \cite{gurarie2011single,prodan2016bulk,volovik2009universe} by:
\begin{lemma}
\label{lem:equivCW}
For $\alpha$ in a global spectral gap, let $G=G_\alpha$ be constructed from $H^\pm(k)$ as described above.
Let $W_\alpha$ and $c_i=c[\Pi_i]$ be defined in \eqref{eq:Walpha} and \eqref{eq:cbd}, respectively. Then we have the relation
\begin{equation}\label{eq:cW}
  W_\alpha = -\dsum_{h_i<\alpha} c_i = \dsum_{h_i>\alpha} c_i.
\end{equation}
\end{lemma}
The proof of the lemma is given in Appendix \ref{sec:appeq}. 

Note that the Chern numbers $c[\Pi_i]$ have been defined for sufficiently smooth projectors without conditions on the energies $h_i(k)$ besides local gaps ensuring the regularity of $k\to\Pi_i(k)$. In contrast, $G_\alpha$ is defined only for $\alpha$ in a global spectral gap. The global spectral gap condition is satisfied in the applications considered here, with a global gap $\alpha\in (-|m|,|m|)$ for \eqref{eq:s22} and two global gaps $\alpha\in (-|f|,0) \cup ( 0,|f|)$ for \eqref{eq:s33}.

The form of the invariant \eqref{eq:Walpha} appears naturally in the subsequent analysis of the bulk-interface correspondence. It also offers a convenient computational tool, in particular \eqref{eq:w3trace2} below that appears in the proof of the above lemma. The formula is used to compute the bulk-difference invariant for the system  \eqref{eq:s33} in Appendix \ref{sec:appeq}.

\section{Bulk-Interface correspondence}
\label{sec:BI}
We return to the computation of the interface invariant $\sigma_I$.

The bulk-interface correspondence stipulates that the amount of asymmetric current, or equivalently the number of topologically protected edge states, is given by the bulk-difference invariant. For the system \eqref{eq:s22}, the difference therefore equals $c_-=2\pi\sigma_I=-\sgn{m}$ in the topologically nontrivial case, whereas for the system \eqref{eq:s33}, it is given by $c_-=2\pi\sigma_I=2\sgn{f}$. In both cases, $\mu(y)$ is a smooth switch function in $\fS[-\mu_,+\mu]$.

Many techniques have been developed to prove or at least build intuition on the correspondence between bulk invariants and the number of edge or interface modes. The edge problem considers bulk Hamiltonians in the plane and their restriction to a half plane with appropriate boundary conditions along the edge. The relation between the bulk invariant and the number of topological edge states is then referred to as a bulk-boundary correspondence \cite{bourne2017k,bourne2018chern,elbau2002equality,Graf2013,graf2018bulk,hatsugai1993chern, kellendonk2004quantization,ludewig2020cobordism,prodan2016bulk}. 

This paper considers a similar problem,  the domain wall problem, where the order parameter $\mu(y)$ transitions, smoothly or not, from one bulk topology to another. The relation between the two bulk invariants, and more precisely the bulk-difference as established in the preceding section, and the number of protected interface modes is the bulk-interface correspondence \cite{fukui2012bulk,essin2011bulk,volovik2009universe}. A brief discussion on these works on spectral flows and spectral asymmetries and their relation to the results of this paper is presented in Appendix \ref{sec:bic}.

The main advantage of deriving a bulk-interface correspondence is that it avoids the explicit spectral decomposition and the search for topologically non-trivial branches of continuous spectrum that we carried out in section \ref{sec:ii}. The computation of $\sigma_I$ is directly obtained from calculations of presumably simpler bulk invariants. 

The main result of this section is to recast the interface conductivity as an appropriate integral of the symbol of the Hamiltonian $H$ that is familiar in index theory \cite{atiyah1975spectral,H-III-SP-94,niemi1984spectral}. This integral  is in the form of a Fedosov-H\"ormander formula and computes the index of a Fredholm operator naturally related to $H$. By an application of the Stokes' theorem, it is also related to the bulk-difference invariant in \eqref{eq:Walpha}. This correspondence  essentially realizes the topological charge conservation of \cite{essin2011bulk,volovik2009universe}.


\medskip

We recall the chain of relations obtained under different assumptions in section \ref{sec:ii}:
\begin{equation}\label{eq:rel}
  2\pi \sigma_I = 2\pi {\rm Tr} \ i[H,P] \varphi'(H) = -\Tr [P,U(H)]U^*(H) =  {\rm Index} \ PU(H)P_{{\rm Ran P}} =:I[H]. \end{equation}
We recall that $\varphi$ is a smooth function increasing from $0$ to $1$ with $\varphi'\geq0$ supported in a spectral gap while $U(H)=e^{i 2\pi \varphi(H)}$. The main objective of the section \ref{sec:icpdo} is to find sufficient conditions on the symbol of $H$ such that \eqref{eq:rel} holds. 

Once $2\pi\sigma_I$ will be guaranteed to be integer-valued and constant over continuous deformations of $H$, we will evaluate the trace defining $\sigma_I$ by applying semiclassical calculus tools similar to those used to derive the index of elliptic pseudo-differential operators on $\Rm^2$ \cite[Chapter 19]{H-III-SP-94} in section \ref{sec:smc}.

\subsection{Interface conductivity and pseudo-differential calculus}
\label{sec:icpdo}

The terminology used in the rest of the paper on pseudo-differential operator (\pdo) and semiclassical pseudo-differential operators (\hpdo) is borrowed from \cite{dimassi1999spectral}; see also \cite{bolte2004semiclassical} for the extension of results to matrix-valued symbols. Appendix \ref{sec:hpdo} summarizes the notation and results on \pdo\ and \hpdo\ we use in this paper. 
Throughout, we use the notation $\aver{X}=\sqrt{1+X^2}$ for $X\in\Rm$ and $\aver{X}^{-\infty}$ as a quantity bounded by $C_N\aver{X}^{-N}$ for every $N\in\Nm$.

\medskip

We now find sufficient criteria ensuring that  $[H,P] \varphi'(H)$ and $[P,U(H)]U^*(H)$ are trace-class.
Let $H={\rm Op}^w(\sigma)$ be a self-adjoint \pdo\ with symbol $\sigma=\sigma(x,y,\xi,\zeta)$ in $S(m_n)$ with $m_n(x,y,\xi,\zeta):=(\aver{x}+\aver{y}+\aver{\xi}+\aver{\zeta})^n$ for $n\in\Nm$. Let $\alpha\in\Rm$ be an energy level. We want to show that $(H-\alpha)^2={\rm Op}^w(\tau_\alpha)$  with symbol $\tau_\alpha=(\sigma-\alpha)\sharp(\sigma-\alpha)$ is well-approximated by a positive-definite operator $G_\alpha={\rm Op}^w(\tilde\tau_\alpha) $.

\begin{proposition}\label{prop:traceclass}
  Let $\alpha\in\Rm$ and $E_0>0$.
  Let $H={\rm Op}^w(\sigma)$ with $\sigma\in S(m_n)$ and $(H-\alpha)^2={\rm Op}^w(\tau_\alpha)$ as above.\\
 We assume the existence of  $G_\alpha={\rm Op}^w(\tilde\tau_\alpha)$ a self-adjoint operator such that $G_\alpha$ has no spectrum in $(-\infty,E_0^2]$ and such that {either}: 
\\[2mm] 
{\rm (H1)} $\tilde\tau_\alpha=\tau_\alpha$ for $y^2+\zeta^2+\xi^2>R^2$; {or}: 
\\[2mm]
{\rm (H2)} $\tau_\alpha-\tilde\tau_\alpha$ in $S(m)$ with  $m=\aver{y}^{-\infty}(\aver{\xi}+\aver{\zeta})^{2n}$ and $(I+(H-\alpha)^2)^{-1}$ has Weyl symbol in $S(m)$ with  $m=(\aver{\zeta}+\aver{\xi})^{-s}$ for $s>0$.
\\[2mm]
Let $\phi$ be a smooth function with compact support in $(-E_0,E_0)$ and $P$ a smooth switch function in $\fS[0,1]$.
  \\[2mm]
Then $[P,\phi(H-\alpha)]$ and $[P,H]\phi(H-\alpha)$ are trace-class operators with traces computed as the integral of their Schwartz kernel along the diagonal.
\end{proposition}
\begin{proof}
   By hypothesis, $G_\alpha$ has no spectrum in $(-\infty,E_0^2+\delta)$ for some $\delta>0$ and we can thus construct a smooth function compactly supported $\psi:[0,\infty)\to\Rm$ such that $\psi(G_\alpha)=0$ and $\psi(E)=1$ for $0\leq E\leq E_0$. This implies that $\psi(\lambda^2) \phi(\lambda)=\phi(\lambda)$ by assumption on $\phi$. 

   By spectral calculus and following a similar construction in \cite{helffer1983calcul}, we have
\[
    \phi(H-\alpha)=\psi(H_1) \phi(H-\alpha) = \big(\psi(H_1)-\psi(H_2)) \phi(H-\alpha),
\] 
with $H_1=(H-\alpha)^2$ and $H_2=G_\alpha$. Therefore, $[P,H]\phi(H-\alpha)$ is trace-class when $[P,H]\psi(H_1)$ is since $\phi(H-\alpha)$ is bounded.  Taking a commutator with $P$ provides localization in $x$ since
\[
  [P,A] =  \chi(x) A (1-\chi(x))-(1-\chi(x)) A \chi(x) ,
\] 
with $\chi(x)$ a multiplication operator with symbol bounded by $\aver{x}^{-\infty}$ for $x<1$ and $1-\chi(x)$ an  operator with symbol bounded by $\aver{x}^{-\infty}$ for $x>-1$. For $A$ a \pdo\ with symbol in $S(m_A)$, then $[P,A]$ has a symbol in $S(m_A\aver{x}^{-\infty})$ by composition of three \pdo; see \cite[Chapter 7]{dimassi1999spectral} and \cite{drouot2019microlocal}.

Let us first prove the results under hypothesis (H1). The Helffer-Sj\"ostrand formula \eqref{eq:hs}  (see also \cite{davies_1995,dimassi1999spectral} and \cite[Eq. (9.11)]{dimassi1999spectral}) gives the expression:
\[
 \psi(H_1) = \psi(H_1)-\psi(H_2) = -\frac1\pi \dint_{\Cm} \bar\partial\tilde\psi (z-H_1)^{-1}(H_1-H_2)(z-H_2)^{-1} d^2z,
\]
where $d^2z=d\lambda d\omega$ for $z=\lambda+i\omega$, and $\tilde\psi$ is an almost analytic extension of $\psi$ as in \eqref{eq:hs}; see also Appendix \ref{sec:hpdo}. The operators $(z-H_k)^{-1}$ have symbols in $S(1)$ bounded by $C|{\rm Im} z|^{-1}$ uniformly on the compact support of $\bar\partial\tilde\psi$, itself satisfying $|\bar\partial\tilde\psi|\leq C_N|{\rm Im} z|^{N}$ for any $N\geq1$, while the symbol of $(H_1-H_2)$ is in $S((\aver{y}+\aver{\xi}+\aver{\zeta})^{-\infty})$. By composition calculus, $\psi(H_1)$ is also a \pdo\ with symbol in  $S((\aver{y}+\aver{\xi}+\aver{\zeta})^{-\infty})$. As a consequence, $[P,H]\psi(H_1)$ has a symbol in $S(\aver{x}^{-\infty}(\aver{y}+\aver{\xi}+\aver{\zeta})^{-\infty})$.

The operator $\phi(H-\alpha)=\phi(H-\alpha)\psi(H_1)$ also has a symbol in $S((\aver{y}+\aver{\xi}+\aver{\zeta})^{-\infty})$. Therefore, $[P,\phi(H-\alpha)]$ has a symbol in $S(\aver{x}^{-\infty}(\aver{y}+\aver{\xi}+\aver{\zeta})^{-\infty})$. Both operators are thus trace-class by \cite[Theorems 9.3\&9.4]{dimassi1999spectral} with traces computed by integrating the symbols in $(x,\xi)$ or equivalently the Schwartz kernels along the diagonal.

Let us now assume hypothesis (H2). Then, for $p>0$,
\[
  \phi(H-\alpha) = \phi(H-\alpha) (I+H_1)^p (I+H_1)^{-p}(\psi(H_1)-\psi(H_2))
\]
with $\phi(H-\alpha) (I+H_1)^p$ a bounded operator and by the Helffer-Sj\"ostrand formula 
\[
 (I+H_1)^{-p}(\psi(H_1)-\psi(H_2))= -\frac1\pi \dint_{\Cm} \bar\partial\tilde\psi (z-H_1)^{-1}(I+H_1)^{-p}(H_1-H_2)(z-H_2)^{-1} dz.
\]
On the compact support of $\bar\partial\tilde\psi$,  the symbols of $(z-H_k)^{-1}$ and $\bar\partial\tilde\psi$ are bounded as described above and the symbol of $(I+H_1)^{-p}(H_1-H_2)$ is in $S(\aver{y}^{-\infty}(\aver{\xi}+\aver{\zeta})^{2n-sp})$. The symbols  of both $[P,\phi(H-\alpha)]$ and $[P,H]\phi(H-\alpha)$ are therefore in $S(\aver{x}^{-\infty}\aver{y}^{-\infty}(\aver{\xi}+\aver{\zeta})^{2n-sp})$ for any $p>0$ and since $s>0$, \cite[Theorems 9.3 \& 9.4]{dimassi1999spectral} apply again. 
\end{proof}

For $\alpha\in\Rm$ and $E_0>0$, let $\varphi$ be a smooth switch function in $\fS[0,1,\alpha-E_0,\alpha+E_0]$ and $U(H)=e^{i2\pi \varphi(H)}$. We prove the relations in \eqref{eq:rel} as follows.
\begin{corollary}\label{cor:rel}
  Let $H$ as in the preceding proposition with either $\phi(H)=W(H)=U(H)-I$ or $\phi(H)=\varphi'(H)$. Let $G_\alpha$ as in the above proposition with no spectrum in $(-\infty,E_0^2]$ and such and either (H1) or (H2) holds. 

Then all terms in \eqref{eq:rel} are defined and equal.
\end{corollary}
\begin{proof}
   The above proposition implies that $[P,U(H)]$ and $[P,H]\varphi'(H)$ are trace-class.  The first equality in \eqref{eq:rel} is a definition. The last equality (not counting the definition of $I[H]$) holds as an application of a Fedosov formula provided that $[P,U(H)]$ is compact; see \cite{B-BulkInterface-2018}. The middle inequality has been shown to hold using a variety of derivations \cite{B-BulkInterface-2018,drouot2019microlocal,elbau2002equality,prodan2016bulk}. Following \cite{B-BulkInterface-2018}, we recall that the above middle equality holds provided $[H,P] \varphi'(H)$ and $[P,U(H)]U^*(H)$ are trace-class. 
\end{proof}

The previous result shows that \eqref{eq:rel} holds when $(H-\alpha)^2$ is well approximated by some $G_\alpha$ such that $\|G_\alpha\|>E_0^2$. We will find sufficient criteria ensuring either (H1) or (H2) for the systems \eqref{eq:s22} and \eqref{eq:s33}. Before doing so, we prove several results showing that the conductivity $\sigma_I$ is stable against changes in $P$, $\varphi'$, and $H$. 
\begin{proposition}\label{prop:stabsigma}
 Let $P_j$ for $j=1,2$ be two smooth switch functions in $\fS[0,1]$ and $\varphi_j$ for $j=1,2$ be two smooth switch functions in $\fS[0,1,\alpha-E_0,\alpha+E_0]$ for $\alpha\in\Rm$ and $E_0>0$. Assume that $H$ satisfies the hypotheses of Corollary \ref{cor:rel}. Then we have
\[
  \sigma_I= {\rm Tr}\ i[P_1,H]\varphi_1'(H) ={\rm Tr}\ i[P_2,H]\varphi_2'(H).
\]
\end{proposition}
\begin{proof}
 Both traces are defined and \eqref{eq:rel} applies thanks to Corollary \ref{cor:rel}. Since $P_j$ for $j=1,2$ have the same behavior at infinity, they belong to the same homotopy class of such functions and we find a continuous family of smooth functions $[1,2]\ni t\mapsto P_t\in\fS[0,1]$ with continuity for instance in the uniform norm. Similarly, we construct $[1,2]\ni t\mapsto \varphi_t$ smooth switch functions in $\fS[0,1,\alpha-E_0,\alpha+E_0]$. Using \eqref{eq:rel}, we know that $2\pi\sigma_I(t)$ is the index of $P_t U_t(H)P_t$ and is therefore an integer for all $t\in [1,2]$. By construction of $P_t$, $t\mapsto  P_t U_\tau(H)P_t$ is continuous in operator norm for any fixed value of $\tau\in[1,2]$. 

Also, $U_t(E)=e^{i2\pi \varphi_t(E)}$ is a smooth function of $t$. This translates, for instance by invoking the Helffer-Sj\"ostrand formula \eqref{eq:hs}, into a continuous map $t\to U_t(H)$ in operator norm. 

As a consequence, $t\mapsto P_tU_t(H)P_t$ is continuous in operator norm and the index of these Fredholm operators is therefore independent of $t$ \cite[Chapter 19]{H-III-SP-94}. Appealing to \eqref{eq:rel} concludes the proof of the proposition.
\end{proof}

We now extend the above results to families of operators $H=H(t)$ for $t\in [0,1]$ such that the invariants in \eqref{eq:rel} are independent of $t$. This allows us to include the presence of perturbations $V(x,y)$ in the Hamiltonian and to vary semiclassical parameters in the next section. We first recall the following result:
\begin{lemma} \label{lem:scale}
  Let $A$ be trace-class operator on $L^2(\Rm^n)$ and $\Lambda$ a linear invertible transform in ${\rm GL}(n,\Rm)$. Then $\Lambda^{-1}A\Lambda$ is also trace-class and
$
  {\rm Tr}\ A = {\rm Tr}\  \Lambda^{-1}A\Lambda.
$
\end{lemma}
The proof is a direct consequence of the cyclicity of the trace. The only such linear transform we use in this paper is the scaling $Y\mapsto y=hY$ for $h>0$ so that $hD_y$ and $\mu(y)$ are pulled back to $D_Y$ and $\mu(hY)$. Recall the notation $I[H]:={\rm Index } PU(H)P$. Let $H={\rm Op}^w(a(x,y,\xi,\zeta))$ and $\tilde H_h={\rm Op}^w(a(x,hy,\xi,h^{-1}\zeta))$. Then the lemma and identities such as $U(\Lambda^{-1}H\Lambda)=\Lambda^{-1}U(H)\Lambda$ imply that $I[H]=I[\tilde H_h]$ when either one is defined. What we need is a different equality. Let $H_h={\rm Op}^w(a(x,y,\xi,h\zeta))$ the semiclassical rescaling (in the variable $y$). We want to show that $I[H_h]=I[H]$. The above lemma states that $I[H_h]=I[\tilde H]$ with $\tilde H={\rm Op}^w(a(x,hy,\xi,\zeta))$. It thus remains to show that transforming a coefficient $\mu(y)$ to $\mu(hy)$ does not modify the invariants in \eqref{eq:rel}. Similarly, let $H_t=H+tV$. We want to show that $I[H_t]$ is independent of $t$. A sufficient condition is as follows.
\begin{proposition}\label{prop:stabt}
  Let $H_t$ be as in corollary \ref{cor:rel} for all $t$ in a connected interval $I\subset\Rm$. Let us assume that $H(t)-H(s)=(t-s)B(s,t)$ for $B(s,t)$ a bounded operator in uniform norm on $L^2(\Rm^2;\Mm_n)$ uniformly in $(s,t)\in I^2$. Then $I[H_t]={\rm Index}\ PU(H_t)P$ is independent of $t\in I$.
\end{proposition}
\begin{proof}
  Let $H_t$ be as above and consider the unitary $U(H_t)$. By the Helffer-Sj\"ostrand formula \eqref{eq:hs}, we have
\[
  U(H_t)-U(H_s) = W(H_t)-W(H_s) = -\frac1\pi \dint_{\Cm^2} \bar\partial \tilde U (z-H_t)^{-1}(H(t)-H(s))(z-H_s)^{-1} d^2z.
\]
By hypothesis, $H(t)-H(s)=(t-s)B(t,s)$ with $B(t,s)$ bounded. Since $(z-H_t)$ is bounded by $|{\rm Im}z|^{-1}$ uniformly in $t$ (at least locally) and $\bar\partial \tilde U\leq C_N |{\rm Im}z|^N$for a well-chosen almost analytic extension $\tilde U$, we find that $U(H_t)-U(H_s)$ is bounded in the operator norm by a constant times $(t-s)$. This ensures that $t\mapsto U(H_t)$ is continuous in the operator norm and hence so is the Fredholm operator $PU(H_t)P$. By continuity of the index of Fredholm operators \cite[Chapter 19]{H-III-SP-94}, the index is independent of $t$.
\end{proof}

The above result addresses the question of stability under perturbation $H_V=H+V$. Assuming that $V$ is bounded and that the index $I[H+tV]$ is defined for $0\leq t\leq 1$, the above proposition directly implies that $I[H]=I[H+V]$. The above result also allows us to address stability in the semiclassical regime. Let $H_h=H_0+\mu(hy)$ with $\mu$ multiplication by a bounded (matrix-valued) function and $h>0$. Let us assume that $\mu(y)$ takes constant values for $y$ large and $-y$ large. Then $\mu(hy)-\mu(ly)=(h-l)y\mu'(\tilde hy)$ for $\tilde h$ between $h$ and $l$. By assumption on $\mu$, $y\mu'(\tilde hy)$ is bounded and proposition \ref{prop:stabt} applies. We collect the above results as:
\begin{corollary}\label{cor:stabt}
  Let $H_{h,t}=H_0+\mu(hy)+tV$ be an operator with $H_0$ independent of $y$, $V$ a bounded multiplication operator, and $\mu(y)$ a smooth domain wall taking constant values outside of a compact interval. Assume that $H_{h,t}$ is as in corollary \ref{cor:rel} for all $h_0\leq h\leq h_1$ and $0\leq t\leq 1$. Then the index $I[H_{h,t}]$ is an integer independent of $h$ and $t$ in these ranges.
\end{corollary}

\medskip

The main difficulty in applying the corollary \ref{cor:rel} is therefore to show the existence of a positive definite operator $G_\alpha$ satisfying either (H1) or (H2). If one such operator $G_\alpha$ may be found for each $H_{h,t}$ in corollary \ref{cor:stabt}, then $I[H_{h,t}]$ is an integer independent of $h$ and $t$. It thus remains to construct such an operator $G_\alpha$. The construction is reasonably explicit in the favorable situation when (H2) is satisfied. 
\begin{proposition}\label{prop:strongelliptic}
  Let $\alpha\in\Rm$ be fixed and $H$ a self-adjoint  unbounded differential operator on $L^2(\Rm^2;\Mm_n)$ such that $(I+(H-\alpha)^2)^{-1}$ is a \pdo\ with symbol in $S^0(m)$ with  $m=(\aver{\zeta}+\aver{\xi})^{-s}$ for $s>0$ and such that $H-\alpha$ has constant coefficients when $y\geq R$ and when $y\leq -R$ for some $R\geq0$. Let $H_\pm-\alpha$ be the constant coefficient operators such that $(H-\alpha)f_\pm=(H_\pm-\alpha)f_\pm$ for any test function $f_+$ supported on $y\geq R$ and $f_-$ supported on $y<-R$. We assume that $H_\pm-\alpha$ have a spectral gap $[-m,m]$ for some $m>0$. Let now $\varphi(H)$ be a smooth switch function in $\fS[0,1,\alpha-m,\alpha+m]$. 

Then there is a differential operator $G_\alpha$ with no spectrum in $(-\infty,m^2]$ and such that hypothesis (H2) in proposition \ref{prop:traceclass} holds. As a consequence, \eqref{eq:rel} holds.
\end{proposition}
\begin{proof}
By assumption, $(I+(H-\alpha)^2)^{-1}$ has a (Weyl) symbol with the required decaying assumption. It remains to construct $G_\alpha$. Let $(\phi_1(y),\phi_2(y))$ be a partition of unity such that $\phi_1^2(y)+\phi_2^2(y)=1$ and such that $\phi_1(y)$ is compactly supported and equal to $1$ on $(-R,R)$. We then construct
\[
  G_\alpha = \phi_2(H-\alpha)^2\phi_2 +m^2\phi_1^2.
\]
The (Weyl) symbols of $(H-\alpha)^2$ and $G_\alpha$ agree for $|y|$ sufficiently large (outside of the support of $\phi_1$). The hypothesis on $\tau_\alpha-\tilde\tau_\alpha$ is therefore clear. Moreover, we find that for any function $f\in L^2(\Rm^2;\Rm^n)$ that
\[
  (f,G_\alpha f) = (f\phi_2,(H-\alpha)^2f\phi_2) + m^2 (f\phi_1,f\phi_1)\geq m^2\|f\|^2
\]
since $(H-\alpha)$ is given by $(H_\pm-\alpha)$ on the support of $f\phi_2$ and these operator have a spectral gap in $[-m,m]$. Here, $(\cdot,\cdot)$ and $\|\cdot\|$ are the inner product and norm on $L^2(\Rm^2)\otimes\Cm^n$. This shows the existence of $G_\alpha$ with the required properties. We deduce from Corollary \ref{cor:rel} that \eqref{eq:rel} holds.
\end{proof}

It is straightforward to apply the above result to the $2\times2$ system \eqref{eq:s22} assuming that $m(y)=\pm m_0$ away from a compact domain and $-|m_0|<\alpha<|m_0|$ so that $E_0=|m_0|-|\alpha|$. The result also clearly holds with $D_y$ replaced by $hD_y$ for any $h>0$. It continues to hold if $H$ is replace by $H+V$ for any compactly supported function $V(x,y)$. Indeed, for any of these operators, we find that $(I+(H-\alpha)^2)^{-1}$ has symbol in $S((\aver{\xi}+\aver{\zeta})^{-2})$.

The $3\times3$ system in \eqref{eq:s33} is more challenging. The reason is the presence of the flat band $E_0(\xi)=0$ observed when $f$ is constant. There is therefore no reason for $(I+(H-\alpha)^2)^{-1}$ to be a negative-order \pdo. Only when the coefficients are sufficiently slowly varying can one expect to leverage the presence of a spectral gap. 

The construction of $G_\alpha$ under hypothesis (H1) is done microlocally and proceeds as follows.
\begin{proposition}\label{prop:Galpha}
  Let $H={\rm Op}^w(\sigma(x,y,\xi,\zeta))$ be a differential operator and assume that for some $\alpha>0$, $(H-\alpha)^2={\rm Op}^w(\tau_\alpha(x,y,\xi,\zeta))$ is such that as a symmetric matrix, $\tau_\alpha(x,y,\xi,\zeta)\geq E_0^2>0$ uniformly in $|(y,\xi,\zeta)|\geq R$ for some $R>0$ and in  $x\in\Rm$.

Then for any $\delta>0$, there exists $G_\alpha={\rm Op}^w(\tilde \tau_\alpha(x,y,\xi,\zeta))$ with $\tilde \tau_\alpha(x,y,\xi,\zeta) \geq (E_0-\delta)^2$ uniformly in $(x,y,\xi,\zeta)$ and $\tilde\tau_\alpha=\tau_\alpha$ when $|(y,\xi,\zeta)|\geq R$.

Let $h>0$ and define $H_h={\rm Op}^w(\sigma(x,y,h\xi,h\zeta))$ as well as   $G_{h,\alpha}={\rm Op}^w(\tilde\tau_\alpha(x,y,h\xi,h\zeta))$.  Then there is $h_0>0$ (sufficiently small) and a constant $C$ independent of $0<h\leq h_0$ such that $\tilde\tau_\alpha(x,y,h\xi,h\zeta)=\tau_\alpha(x,y,h\xi,h\zeta)$ when $|(y,\xi,\zeta)|\geq R/h$ and $G_{h,\alpha}\geq F^2:=(E_0-\delta-Ch_0)^2>0$. 

Let $\varphi\in\fS[0,1,\alpha-F,\alpha+F]$ and define $\tilde H_h={\rm Op}^w(\sigma(hx,hy,\xi,\zeta))$ for $0<h\leq h_0$. Then hypothesis (H1) in Proposition \eqref{prop:traceclass} holds for $H_h$ and Lemma \ref{lem:scale} implies that \eqref{eq:rel} holds for $H$ replaced by $\tilde H_h$.
\end{proposition}
This result thus states that \eqref{eq:rel} holds for Hamiltonians $\tilde H_h={\rm Op}^w(\sigma(hx,hy,\xi,\zeta))$ with spatial coefficients of the form $c(hx,hy)$ for $h$ sufficiently small.
\begin{proof}
  By hypothesis, we can construct a matrix-valued symbol  $\tilde\tau_\alpha=\tau_\alpha\geq E_0^2$ when $|(y,\xi,\zeta)|\geq R$. This symbol can then be extended to $\Rm^4$ smoothly such that $\tilde\tau_\alpha\geq (E_0-\delta)^2$. The positivity of the symbol is not sufficient to ensure that the operator $G_\alpha={\rm Op}^w(\tilde \tau_\alpha(x,y,\xi,\zeta))$ itself is positive. 
Replacing $H$ by $H_h$, with $(H_h-\alpha)^2={\rm Op}^w(\tau_\alpha(x,y,h\xi,h\zeta))$, we observe that $\tilde\tau_\alpha(x,y,h\xi,h\zeta)=\tau_\alpha(x,y,h\xi,h\zeta)$ when $|(y,h\xi,h\zeta)|\geq R$ and hence certainly when $|(y,\xi,\zeta)|\geq R/h$. Moreover, we still have that $\tilde\tau_\alpha(x,y,h\xi,h\zeta)\geq  (E_0-\delta)^2$.  The G\aa rding inequality \eqref{eq:garding} allows us to conclude that $G_{h,\alpha}\geq (E_0-\delta-Ch)^2$ and hence is bounded below by a positive constant when $h$ is sufficiently small.

We therefore obtain that (H1) in Proposition \ref{prop:traceclass} holds for $H_h$ and hence that the index $I[{\rm Op}^w(\sigma(x,y,h\xi,h\zeta))]$ is defined by corollary \ref{cor:rel}. We now  invoke Lemma \ref{lem:scale} to observe that $I[{\rm Op}^w(\sigma(x,y,h\xi,h\zeta))]=I[{\rm Op}^w(\sigma(hx,hy,\xi,\zeta))]$, which proves that \eqref{eq:rel} holds for $H$ replaced by $\tilde H_h$.
\end{proof}


\medskip
In the remainder of this section, we apply the above result to the system \eqref{eq:s33}
\begin{equation}\label{eq:H33}
  H = D_x\gamma_1 + D_y \gamma_4 - f(y) \gamma_7 + V(x,y) 
\end{equation}
with $f(y)$ a smooth switching function in $\fS[-f_0,f_0]$ and  $h>0$. We assume that $V(x,y)$ is multiplication by a smooth compactly supported function (still called $V(x,y)$). Let $\tau_\alpha$ be the symbol of $(H-\alpha)^2$ for $\alpha\in(-f_0,0|)\cup(0,f_0)$. $H_\alpha$ has a bulk gap given by $E_0={\rm min}(|\alpha|,f_0-|\alpha|)$. Our objective is to find constraints on $f$ and $V$ such that the matrix $\tau_\alpha(x,y,\xi,\zeta)$ is bounded uniformly below by a positive constant when $|(y,\xi,\zeta)|\geq R$. The construction of the operator $G_\alpha$ then proceeds as in the above proposition.

Let $(y,\xi,\zeta)=\rho(y_0,\xi_0,\zeta_0)$ with $|(y_0,\xi_0,\zeta_0)|=1$ and the scaling factor $\rho\in\Rm_+$. Consider $|y_0|^2\geq\frac12$ first. Then for $\rho$ sufficiently large, the symbol of $(H_{h,V}-\alpha)^2$ becomes constant with eigenvalues bounded below by ${\rm min}(|\alpha|,f_0-|\alpha|)$ as expected.

The region $|y_0|^2\leq\frac12$ is more challenging. We have $(\xi,\zeta)\to\infty$ as $\rho\to\infty$ in that sector. However, only two out of three eigenvalues of the symbol of  $(H-\alpha)^2$ are bounded below now. The third eigenvalue remains bounded and large fluctuations in $V(x,y)$ or $f(y)$ may force it to vanish or become arbitrarily small. In such a case, we cannot construct a positive definite approximation of $(H-\alpha)^2$. Constraints on $V$ and $f$ thus need to be imposed. The Hermitian operator  $(H-\alpha)^2$ is given explicitly by
\[
\left(\begin{matrix}  W_{11}^2+ \delta_x \delta_x^* + \delta_y \delta_y^* & W_{11} \delta_x+\delta_x W_{22} + \delta_y g^* & W_{11}\delta_y+\delta_xg+\delta_y W_{33}  \\
*  & \delta_x^* \delta_x + W_{22}^2 + |g|^2  & \delta_x^*\delta_y + W_{22}g + g W_{33}  \\ * & * & \delta_y^*\delta_y+|g|^2+W_{33}^2 \end{matrix} \right),
\]
where we have defined $W=V-\alpha$ for $V=(V_{i,j})_{1\leq i,j\leq3}$, $\delta_x=D_x+V_{12}$, $\delta_y=D_y+V_{13}$, and $g=if_h+V_{23}$.  The full Weyl symbol of $(H-\alpha)^2$ is obtained by using identities implying that the symbol of $D_y g(y)$ is  $\zeta g(y) -\frac i2 \partial_yg(y)$ and that of $g(y)D_y$ is  $\zeta g(y) +\frac i2 \partial_yg(y)$.
We wish to show that the eigenvalues $\lambda_{1,2,3}$ of the above symbol are positive independently of the scaling factor $\rho\geq R$ for $(y,\xi,\zeta)=\rho(y_0,\xi_0,\zeta_0)$ and $|(\xi_0,\zeta_0)|^2\geq\frac12$.

Let us assume that $V_{12}$ and $V_{13}$ are real-valued and $V_{23}$ is purely imaginary. Assuming constant coefficients (independent of $(x,y)$), the symbol of $(H-\alpha)^2$ is the square of
\[
\left(\begin{matrix}  W_{1} & \tilde\xi & \tilde\zeta \\ \tilde\xi & W_{2} & g \\ \tilde\zeta & -g & W_3 \end{matrix} \right),
\]
where $\tilde\xi$ and $\tilde\zeta$ are the symbols of $\delta_x$ and $\delta_y$, respectively, and $W_j$ stands for $W_{jj}$. Therefore, $|(\tilde\xi,\tilde\zeta)|\to\infty$ as $\rho\to\infty$. The invariants of the above matrix satisfy
\[
  \lambda_1\lambda_2\lambda_3 = -W_3\tilde\xi^2-W_2\tilde\zeta^2 + O(1),\quad \sum_{i\not=j} \lambda_i\lambda_j= -(\tilde\xi^2+ \tilde\zeta^2) + O(1),\quad \lambda_1+\lambda_2+\lambda_3 = O(1),
\]  
where $O(1)$ means independent of $\rho$. This implies that  two eigenvalues are asymptotically equal to $\pm\rho$ while the third eigenvalue is asymptotically given by $\frac{W_3\xi^2+W_2\zeta^2}{\rho^2}$. When $V_{22}=V_{33}=0$, the latter term equals $-\alpha\not=0$. In order for that term not to vanish, we need $V_{22}$ and $V_{33}$ to be sufficiently small compared to $\alpha$. Large variations in $V_{22}$ and $V_{33}$ prevent the approximation of $(H-\alpha)^2$ by a positive definite operator. Note that we do not have any constraint on $V_{11}$. Similar calculations show that the imaginary part of $V_{12}$ and $V_{13}$ and the real part of $V_{23}$ also need to be small for the above eigenvalues to remain bounded away from $0$. 

The full symbol of $(H-\alpha)^2$ is given by the square of the above matrix plus contributions coming from differentiations of the coefficient such as $\partial_{x,y} g$ or $\partial_{x,y}V_{i,j}$. For instance, the term $\partial_y g$, which involves $f'(y)$, modifies the determinant and the sum of squares of eigenvalues by a quantity of order $(\xi^2+\zeta^2)\zeta^2$, which is leading-order in $\rho^4$. Such terms have to be sufficiently small for the symbol of $(H-\alpha)^2$ to have positive eigenvalues uniformly. 

Let $\eps>0$ and assume the coefficients in \eqref{eq:H33} are of the form $f:=f(\eps y)$ scalar-valued and $V:=V(\eps y)$ Hermitian-valued.  Let $\sigma_\eps$ be the Weyl symbol of $H={\rm Op}^w(\sigma_\eps)$. We introduce the following hypothesis:
\\[2mm]
(H$3\times 3$): The coefficients $V_{22}$, $V_{33}$, ${\rm Im} V_{12}$, ${\rm Im} V_{13}$ and ${\rm Re} V_{23}$ are sufficient small in $S(1)$ and $\eps$ is sufficiently small. 
\\[2mm]
Under (H$3\times 3$), the above calculations show that the symbol of $(H+V-\alpha)^2$ has eigenvalues bounded away from $0$ uniformly in $|(y,\xi,\zeta)|>R$ for $R$ sufficiently large. Therefore,  $H={\rm Op}^w(\sigma_\eps)$  satisfies the hypotheses of Proposition \ref{prop:Galpha}. Let $F=E_0-\delta$ for $\delta>0$ and consider a density of states built on $\varphi\in\fS[0,1,\alpha-F,\alpha+F]$. Then Proposition \ref{prop:Galpha} implies that \eqref{eq:rel} holds provided that $\eps$ is further reduced to $h_0\eps$. Relabeling $\eps$ the small term $\eps h_0$, we thus deduce that \eqref{eq:rel}  holds for any Hamiltonian
$H_\eps = D_x\gamma_1 + D_y \gamma_4 - f(\eps y) \gamma_7 + V(\eps x,\eps y)$
satisfying (H$3\times3$).

This shows in particular that the interface conductivity in \eqref{eq:rel} is well defined provided that $f(y)$ is a sufficiently slowly varying domain wall. We saw in section \ref{sec:ii} that \eqref{eq:rel} also held for $f(y)=f_0\sgn{y}$. However, this result came from a sufficiently explicit knowledge of the spectral decomposition of $H[\mu(y)]$. The result we just established applies to any smooth $f(y)$ that has sufficiently slow variations. 

Numerical simulations in \cite{QB-NUMTI-2021} show that the value of $\sigma_I$ is indeed quite stable against smooth perturbations in  $V_{11}$, $V_{12}$, $V_{13}$ and ${\rm Im}V_{23}$ while even small variations in the other components of $V$ rapidly destabilize it.

%
%
%
\subsection{Fedosov-H\"ormander formula and Bulk-Interface correspondence}
\label{sec:smc}
%
%
%

This section presents a general bulk-interface correspondence for the family of differential operators
\begin{equation}\label{eq:Hh}
   H_h = D_x \gamma_1 + \gamma_2(hD_y,y)
\end{equation}
for $\gamma_1$ a Hermitian matrix in $\Mm(\Cm^n)$ and $\gamma_2(D_y,y)$ a differential operator taking values in $\Mm(\Cm^n)$. We also define $H_V=H_h+V$ a perturbed operator. We set $H:=H_1$.
In the applications considered in this paper, $\gamma_2(hD_y,y)=hD_y\gamma_2+\mu(y)\gamma_3$ but the specific structure is not important so long as (h2) below holds. This class of models finds applications outside of the systems considered here, for instance for the $n-$replica models that appear in the analysis of Floquet topological insulators \cite{BM-FTI-21}. Generalizations to a larger class of models mixing the $x$ and $y$ derivatives and involving higher-order derivatives are worked out in \cite{QB-NUMTI-2021}.

We assume that \eqref{eq:rel} applies uniformly in $0<h\leq1$ and that all traces involved there are computable as integrals along of the diagonal of the Schwartz kernel of the trace-class operators. This holds thanks to Proposition \ref{prop:strongelliptic} for the system \eqref{eq:s22} and thanks to Proposition \ref{prop:Galpha} for the system \eqref{eq:s33} provided that $\mu(y)=f(\eps y)$ with $\eps$ sufficiently small. The preceding section also indicated which perturbations $V$ were allowed for the two systems \eqref{eq:s22} and \eqref{eq:s33}. 

\medskip

Let $\alpha\in\Rm$ fixed. Assuming $\alpha$ is in a bulk band-gap, we assigned two invariants so far. One is the conductivity $\sigma_I$ in \eqref{eq:rel}. The other one is the bulk-difference invariant $W_\alpha$ and its relation to Chern numbers in \eqref{eq:cW}. We now introduce a third invariant as the index of the operator $H_h-\alpha-i\nu(x)$ for well-chosen domain walls $\nu(x)$. The objective of the bulk-interface correspondence is to show that these invariants are all equal and different calculations of the same topological charge.

\medskip

The operator $H_h={\rm Op}^w_h(\sigma)$ with $S^0(m)\ni \sigma = \xi\gamma_1+\gamma_2(\zeta,y)$ for the order $m$ that makes $H_h$ self-adjoint as an unbounded operator on $L^2(\Rm^2)\otimes\Cm^n$.  We define $H_h[\xi]$ in \eqref{eq:Hxi} as the Fourier transform $x\to\xi$ of $H_h$, which is invariant by translation along the $x$ axis.  We collect our assumptions as follows:
\\[2mm]  
(h1)  $2\pi\sigma_I$ is an integer independent of $0<h\leq1$  and $i[H_h,P] \varphi'(H_h)$ is a trace-class operator with trace given by the diagonal integral of its Schwartz kernel. Here, $\varphi$ is a smooth switch function in $\fS[0,1,\alpha-E_0,\alpha+E_0]$ for $\alpha\in\Rm$ and $E_0>0$ and $P$ is a multiplication operator by a smooth function $\chi$ in $\fS[0,1,x_0-\beta,x_0+\beta]$ for $\beta>0$. Then $[H_h,P]=-i\chi'(x)\gamma_1$ with $\chi'(x)$ compactly supported.
\\[2mm] 
(h2) Let $z=\alpha+i\omega$ and $\sigma_z(y,\zeta;\xi)$ be the semiclassical symbol of $z-H_h[\xi]={\rm Op}_h(\sigma_z)$ in the variables $(y,\zeta)$ with $(\alpha,\omega,\xi)$ seen as parameters. Let $\sigma_\alpha(\omega,y,\xi,\zeta)=\sigma_{\alpha+i\omega}(y,\zeta;\xi)$ be the Weyl symbol at $h=1$ of $\alpha+i\omega-H$.
We assume that $\sigma_\alpha$ is invertible and $\sigma_\alpha^{-1}d_{\omega,y,\xi,\zeta}\sigma_\alpha$ is uniformly bounded (one-form with bounded coefficients) for $\omega^2+y^2+\xi^2+\zeta^2\geq R^2$ sufficiently large. Finally, we assume that $\sigma_\alpha^{-1}d_{\omega,y,\xi,\zeta}\sigma_\alpha\to0$ as $|\zeta|\to\infty$.

\begin{remark}
 (i) Condition (h1) corresponds to what was proved in section \ref{sec:icpdo}. Condition (h2) is very similar to what is necessary to construct bulk-difference invariants. Indeed by the same diagonalization as in \eqref{eq:diagH}, we find 
\[
  \sigma_\alpha = \alpha+ix - \sum_{j=1}^n h_j(\xi,y,\zeta) \Pi_i(\xi,y,\zeta),\qquad 
  \sigma_\alpha^{-1} = \sum_{j=1}^n (\alpha+ix-h_j(\xi,y,\zeta))^{-1} \Pi_i(\xi,y,\zeta)
\]
where we see that $\sigma_\alpha$ is invertible if all $\alpha+ix-h_j(\xi,y,\zeta)\not=0$ for $(x,\xi,y,\zeta)$ outside of a compact set.  The case $x=0$  implies that $\alpha$ lives within the bulk gap of the system.

(ii) In this paper, we considered two forms of mass terms $\mu(y)$: either bounded domain walls or $\lambda y$. The same applies to the confining term $x$ above, which may be replaced by any domain wall $\nu(x) \in \fS[-|\alpha|-\delta,|\alpha|+\delta]$ for $\delta>0$ or a term of the form $\nu(x) = |\alpha| \arctan x$ with a gap sufficiently large to include the energy level $\alpha$. The above invertibility condition then takes the form $\alpha+i\nu(x)-h_j(\xi,y,\zeta)\not=0$ outside of a compact set. 

(iii)  Let $\nu(x)=x$ or $\nu(x) \in \fS[-|\alpha|-\delta,|\alpha|+\delta]$ for $\delta>0$.
For  system \eqref{eq:s22}, condition (h1) is $\alpha+i\nu(x)\not=\pm\sqrt{\xi^2+\zeta^2+m^2(y)}$. For $m(y)\in \fS[-|m_0|,|m_0|]$, this condition is satisfied when $|\alpha|<|m_0|$ (but not for $|\alpha|>|m_0|$).  For the geophysical system \eqref{eq:s33}, the same conditions hold with $m(y)$ replaced by $f(y)$ with the additional constraint in (h1)  that $\alpha+i\nu(x)\not=0$, i.e, $\alpha\not=0$ when $\nu(x)=0$.
We verify that these conditions are satisfied for $\alpha$ in a bulk band gap, i.e., $\alpha\in (-|f_0|,0)\cup  (0,|f_0|)$ for smooth $f(y)\in\fS[-f_0,f_0]$ and $\alpha\not=0$ for $f(y)=\lambda y$, $\lambda\not=0$.
\end{remark}

\medskip

While the main objective of the bulk interface correspondence in Theorem \ref{thm:bic} below is to relate the interface conductivity $\sigma_I$ to the bulk-difference invariant $W_\alpha$, it turns out that both invariants are connected by a third one, a Fredholm operator naturally related to $H_h$ whose index captures the topological charge. The Fredholm operator is
\begin{equation}
\label{eq:F}
  F_h = H_h - \alpha -i \nu(x)  
\end{equation}
as alluded to in the above remark with some choices for $\nu(x)$ given there. We set $F:=F_1$.

We assume $\nu(x)$ bounded to simplify the presentation and define $F_h$ as an unbounded operator from ${\mathcal D}(H)$  to $L^2(\Rm^2)\otimes \Cm^n$ where the domain of $H_h$ (which makes it self-adjoint as an unbounded operator on $L^2(\Rm^2)\otimes \Cm^n$) is independent of $h$. There is in fact no reason for any operator of the form \eqref{eq:F} to be Fredholm, that is an operator that admits left and right inverses modulo compact operators. In the preceding section, we introduced two hypotheses (H1) and (H2) on the symbol of $H$. We state corresponding hypotheses on that of $F_h$ when $\nu(x)\in \fS[-\delta-|\alpha|,|\alpha|+\delta]$ to simplify; similar results hold for $\nu(x)=x$ and unbounded domain walls $\mu(y)$ after appropriate changes of the domain ${\mathcal D}(F)$, which we do not consider here in detail. 
Note that we can always choose $\nu\in\fS[-\delta-|\alpha|,|\alpha|+\delta]$ such that $\nu(x)=x$ for $|x|\leq R$ so that the Weyl symbol of $F_h$ is given by $-\sigma_\alpha(x,y,\xi,\zeta)$ defined in (h2) when restricted to the vicinity of the sphere $x^2+y^2+\xi^2+\zeta^2=R^2$. Only those values appear in the Fedosov-H\"ormander formula \eqref{eq:bic0} below.  The hypotheses are:
\\[2mm]
\noindent (H1') Let $F_h$ be as in \eqref{eq:F} and assume that (h2) holds for $0< h\leq h_0$. 
\\[2mm]
\noindent (H2') Let $m=(1+\xi^2+\zeta^2)^{-s}$ for $s>0$. Assume $(I+F^*F)^{-1}={\rm Op}^w(a)$ with $a\in S(m)$ and (h2) holds. 
\\[2mm]
Hypothesis (h2) may be seen as an ellipticity condition at infinity. In the favorable case (H2') of a regularizing operator, then (h2) is sufficient to obtain a Fredholm operator. For systems such as \eqref{eq:s33}, where regularization does not occur for all components, a smallness condition in (H1') is necessary to ensure the Fredholm structure, in parallel to (H1) in section \ref{sec:icpdo}.

\begin{proposition}\label{prop:fredholm}
 Let $F_h$ be the differential operator defined in \eqref{eq:F}.  Let us assume that: either (H1') holds with $h_0$ sufficiently small;  or (H2') holds. Then $F_h$ is a Fredholm operator from ${\mathcal D}(H)$  to $L^2(\Rm^2)\otimes \Cm^n$.
\end{proposition}
\begin{proof} 
Assume (H2').  By assumptions, the differential operator $F^*F$ is a constant-coefficient operator outside of a compact domain $\Omega$, which is invertible. Let $G=F^*F+\chi^2(x,y)$ with $\chi^2\in C^\infty_0(\Rm^2)$ with $\chi^2=1$ on $\Omega$. Then $G$ is positive definite and hence invertible with $G^{-1}\in S(m)$ by assumption. We thus have $G^{-1}F^*F=I+G^{-1}\chi^2$ and we find that $G^{-1}\chi^2$ is compact since $\chi^2$ is compactly supported in space and $G^{-1}$ is smoothing. We have thus constructed a left inverse up to a compact operator. We can similarly construct a right inverse and conclude that $F$ is Fredholm \cite[Corollary 19.1.9]{H-III-SP-94}.


Let us now consider the case (H1'). We can no longer invoke a compactness argument alone and thus need a smallness condition to obtain the Fredholm property. Let $F_h={\rm Op}_h^w(f)$. Let $\chi+\psi=1$ with $\psi(x,y,\xi,\zeta)\in C^\infty_0(\Rm^4)$ compactly supported and equal to $1$ on a sufficiently large domain that  $g=\chi f^{-1}\in S(m^{-1})$ is defined. Let $G={\rm Op}_h^w(g)$. By semiclassical calculus \cite[Chapters 7\&8]{dimassi1999spectral}, we find that $GF={\rm Op}^w_h(\chi)+h{\rm Op}^w_h(r_h)$ with $r_h\in S^0(1)$. Therefore, for $0\leq h\leq h_0$, $I+h{\rm Op}^w_h(r_h)$ is invertible and $(I+h{\rm Op}^w_h(r_h))^{-1}GF=I-{\rm Op}^w_h(\psi)$. However, ${\rm Op}^w_h(\psi)$ is a compact operator so that $F$ admits a left inverse up to a compact perturbation. We construct a right-inverse similarly and conclude that $F_h$ is Fredholm for $0< h\leq h_0$.  
\end{proof}
It remains to show that (H1') applies to \eqref{eq:s33} when the coefficients (including appropriate perturbations $V$) vary sufficiently slowly and that (H2') applies to \eqref{eq:s22}. This is done along the lines of the derivation in the preceding section. We leave the details to the reader. That $F$ is a Fredholm operator extends to unbounded domain walls as in \cite[Chapter 19]{H-III-SP-94} as well as operators with coefficients whose derivatives vanish at infinity; see for instance \cite{grushin1970pseudodifferential}. 

\medskip

After these preliminary results and hypotheses, we state the main result of this section:
\begin{theorem}[Bulk-interface correspondence.] \label{thm:bic}
  Let $H$ be given by \eqref{eq:Hh} with $h=1$. Assume (h1)-(h2) above for $\alpha\in\Rm$. Then
\begin{equation}\label{eq:bic0}
   2\pi \sigma_I = W_\alpha = \dfrac{1}{24\pi^2} \dint_\Sigma {\rm tr}  (\sigma_\alpha^{-1}d\sigma_\alpha)^{\wedge 3} = {\rm Index} (F),
\end{equation}
where $\Sigma=\{(\omega,y,\xi,\zeta)\in\Rm^4;\ \omega^2+y^2+\xi^2+\zeta^2=R^2\}$ with the orientation on $\Rm^4$ given by $d\omega\wedge d\xi \wedge dy \wedge d\zeta>0$.
\end{theorem}
We recall that $\sigma_I$ is the interface invariant while $W_\alpha$ is the bulk-difference invariant.
The last index is the $L^2-$index of the unbounded Fredholm operator $F=H-\alpha-i\nu(x)$ on $L^2(\Rm^2)\otimes\Cm^n$ with (Weyl) symbol at $h=1$ given by $-\sigma_\alpha(\nu(x),y,\xi,\omega)$. The invariance properties of the integral in \eqref{eq:bic} show that it takes the same values for $\sigma_\alpha(\nu(x),y,\xi,\omega)$ and $\sigma_\alpha(x,y,\xi,\omega)$ \cite[Chapter 19.3]{H-III-SP-94}.

The rest of this section is devoted to a proof of this result and some corollaries.
We will first need properties of the symbol of (semiclassical) resolvents, which we summarize as follows. 
\begin{lemma}\label{lem:rz}
   Let $H_h={\rm Op}_h(a)$ with $a\in S^0(m)$. Let $z=\lambda+i\omega\in\Cm$ with $\omega\not=0$. Then $(z-H_h)^{-1}$ is a bounded operator and there exists an analytic function $z\to r_z=r_z(y,\zeta;h)$ such that $(z-H_h)^{-1}={\rm Op}_h(r_z)$. The coefficient $r_z\in S^0(1)$ satisfies
\[
  |r_z| \leq C | \omega|^{-1} {\rm max} (1,h^{\frac32}|\omega|^{-3})
\]
for a constant $C$ independent of $z\in Z\subset\Cm$ a compact set and $0<h\leq1$.


\end{lemma}

\begin{proof}
  We follow \cite[Chapters  7\&8]{dimassi1999spectral} and \cite{bolte2004semiclassical} for the extension to matrix-valued operators. $z-H_h={\rm Op}_h(z-a)$ with $\omega\not=0$ so that $z-a$ invertible. The Beals' criterion applied in \cite[(8.10)]{dimassi1999spectral} shows that $(z-H_h)^{-1}={\rm Op}_h (r_z)$ for $r_z\in S^0(1)$. This bound is uniform on any compact domain away from $\omega=0$.  \cite[Proposition 8.6]{dimassi1999spectral} shows that $|r|$ is bounded by the maximum of $|\omega|^{-1}$ and $h^{\frac{2d+1}2}|\omega|^{-(2d+2)}$ where dimension $d=1$ in our applications.
Finally, $r_z$ may be written as an appropriate transform of the Schwartz kernel of $(z-H_h)^{-1}$ \cite[Chapters  7]{dimassi1999spectral}, which is analytic since $\bar\partial (z-H_h)^{-1}=0$. This implies the analyticity of $z\to r_z$ on $Z$. 
\end{proof}

The above lemma is applied to the operator $H_h[\xi]$. For $\Cm\ni z=\lambda+i\omega$ with $\omega\not=0$, we define $r_z$ such that  
\begin{equation}\label{eq:rz}
(z-H_h[\xi])^{-1}={\rm Op}_h(r_z)
\end{equation}
with $r_z=r_z(y,\zeta;\xi;h)\in S^0(1)$ (uniformly in $\xi$) a symbol in the variables $(y,\zeta)$ with $(z,\xi)$ as parameters.

The bulk of the proof of the above theorem is to relate $2\pi\sigma_I$ to the symbols $\sigma_z$ and $\sigma_\alpha$ as follows. 
\begin{proposition}\label{prop:sa} 
Under hypotheses (h1)-(h2), we have 
\begin{equation}
\label{eq:sa0}
  2\pi \sigma_I  = \dfrac{i}{8\pi^2} \lim_{M_{1,2,3}\to\infty} \dint_{\Rm\times \partial \R} \eps_{ijk} \tr (\sigma_\alpha^{-1} \partial_i \sigma_\alpha  \sigma_z^{-1} \partial_j\sigma_\alpha  \sigma_\alpha^{-1}\nu_k)  d\omega d\Sigma,
\end{equation}
where the variables labeled $1,2,3$ are identified with $\xi,\zeta,y$; $\nu_k$ is the $k$th component of the outward unit normal to the boundary of the rectangle $\R=\prod_{l=1}^3(-M_l,M_l)$; and $d\Sigma$ is the Euclidean measure on that boundary. Also, $\eps_{ijk}$ is the totally antisymmetric tensor such that $\eps_{111}=1$. 
\end{proposition}

By hypothesis (h2), the matrix $\sigma^{-1}_\alpha$ is indeed defined when the coefficients $M_\alpha$ are sufficiently large.  The above result shows that the conductivity, and hence all quantities appearing in \eqref{eq:rel} may be written in terms of the symbol $-\sigma_\alpha$ of $H-\alpha-i\omega$. 

\begin{proof}
 Recall that $[H_h,P]=-i\chi'(x)\gamma_1$.  By hypothesis (h1), traces may be computed as (Lebesgue) integrals along diagonals. Integrating in $x$ and denoting by ${\rm Tr}_y$ the matrix trace and integration of the Schwartz kernel along the diagonal in the variable $y$, we have (by Fubini)
\[
2\pi \sigma_I = 2\pi {\rm Tr} \chi'(x) \gamma_1 \varphi'(H_h) = 2\pi \dint_\Rm {\rm Tr}_y  \chi'(x) \gamma_1 (\varphi'(H_h))(x,x) dx.
\]
Since $\varphi'(H_h)$ is invariant by translation, then $\varphi'(H_h)(x,x)=\varphi'(H_h)(0,0)$ and up to $2\pi$ equals the integral over the dual variable $\xi$ so that (dropping $y$ in ${\rm Tr}\equiv {\rm Tr}_y$ moving forward), 
\[
  2\pi \sigma_I = {\rm Tr}\dint_{\Rm} \gamma_1 \varphi'(H_h[\xi]) d\xi = \dint_{\Rm}  {\rm Tr} \gamma_1 \varphi'(H_h[\xi]) d\xi.
\]
Note that $\gamma_1=\partial_\xi H_h[\xi]$ and the above integrand is formally given by $\partial_\xi \varphi(H_h[\xi])d\xi=d\varphi(H_h[\xi])$ so that $2\pi\sigma_I$ depends only on the behavior of $H_h[\xi]$ for $|\xi|$ large. However, $\varphi$, unlike $\varphi'$, is not compactly supported and the above expression is more suitable for the sequel.

We now use the Helffer-Sj\"ostrand formula and  an almost analytic extension  $\widetilde{\varphi'}(z)$ of $\varphi'(\lambda)$ to write $\varphi'(H_h[\xi])$ as a superposition of resolvent operators of $H_h[\xi]$; see \eqref{eq:hs} and  \cite{davies_1995,dimassi1999spectral}
\[
 2\pi \sigma_I = {\rm Tr} \frac{-1}\pi\dint_{\Cm\times\Rm } \gamma_1 \bar\partial \widetilde{\varphi'}(z) (z-H_h[\xi])^{-1} d^2z d\xi =  \frac{-1}\pi{\rm Tr}  \dint_{\Cm\times\Rm } \bpar \widetilde{\varphi'} (z) \gamma_1 {\rm Op}_h(r_z) d^2z d\xi.
\]
Here, $d^2z=d\lambda d\omega$ for $z=\lambda+i\omega$ and $\bar\partial=\frac12\partial_\lambda+\frac i2\partial_\omega$. The resolvent has a semiclassical symbol $r_z\in S^0(1)$ defined in \eqref{eq:rz} with properties collected in Lemma \ref{lem:rz}. The above may also be recast by Fubini, using hypothesis (h1), as
$\frac{-1}\pi\int_\Rm {\rm Tr} \ {\rm Op}_h \big( \int_{\Cm} \bar\partial \widetilde{\varphi'} (z) \gamma_1r_z d^2z \big) d\xi$.
We thus apply \cite[Theorem 9.4]{dimassi1999spectral} to write the trace, which may be written as an integral of its Schwartz kernel, or equivalently for a semi-classical operator in terms of its symbol as:
\begin{equation}\label{eq:condtrace}
 2\pi\sigma_I =\frac{-1}{2\pi^2 h}\dint_{\Cm\times\Rm^3} \bar\partial \widetilde{\varphi'}(z)\ \tr\  \gamma_1 r_z(y,\zeta;\xi;h)  d^2z d\R.
\end{equation}
Here, tr is matrix trace and $d\R:=dy d\zeta d\xi$.

By construction of the almost analytic extension $\widetilde{\varphi'}(z)$,  the above integration in $z$ is performed over a bounded domain $Z$. Moreover, since $|r_z|$ is bounded by $C|\omega|^{-4}$ uniformly in $h$, which is compensated by $|\bpar \tilde\varphi'(z)|\leq C |\omega|^4$, then for each $0<\eps\to0$, we find $0\leq\delta=\delta(\eps)\to0$ such that the above integral over $|\omega|\leq\delta$ is bounded by $\eps$ uniformly in $0<h$. Let us call $Z_\delta= Z\cap\{|\omega|\geq\delta\}$.  We wish to estimate
\[
  I_\delta = \frac{-1}{2\pi^2 h}\dint_{Z_\delta \times\Rm^3} \bar\partial \widetilde{\varphi'}(z)\ \tr  \gamma_1 r_z(y,\zeta;\xi;h)  d^2z d\R.
\] 

The advantage of the above truncation is that $r_z\in S^0(1)$ uniformly in $z\in Z_\delta$ now since $\delta$ is independent of $h$.
The function $z\to r_z$ is analytic away from $\omega=0$ so that $\partial_\lambda r +i\partial_\omega r=0$. Moreover $r_z\to0$ as $|\omega|\to\infty$ since $\|(z-H_h[\xi])\|^{-1}\leq |\omega|^{-1}$. We can therefore write for $\omega>0$
\[
 r(z) = -\dint_0^\infty \partial_\omega r_{z+it} dt = -i \dint_0^\infty \partial_\lambda r_{z+it}dt.
\]
The domain $\omega<0$ is treated similarly by observing that $r(\bar z)=\bar r(z)$. 
We now focus on $\partial_\lambda r_z(y,\zeta;\xi;h)$ for $|\omega|\geq\delta$. Since $r_{z+it}\in S^0(1)$ uniformly, this allows us to obtain asymptotic expansions in $h$ that hold uniformly in $|\omega|\geq\delta$. We then perform an integration in $(\xi,y,\zeta)$ to display the topological properties of the resolvent and obtain a function that no longer has any singularity at $\omega=0$ thanks to the assumption that $\sigma_\alpha$ is invertible outside of a large sphere in $(\omega,y,\xi,\zeta)$. 

By analyticity, $\partial_\lambda r_z=\partial_z r_z$. We wish to perform an expansion of the integral in $(\xi,y,\zeta)$ of $\gamma_1 \partial_\lambda r_z$. At the operator level, we find $\partial_z(z-H_h)^{-1}=-(z-H_h)^{-2}$, which at the level of symbols translates into $\partial_\lambda r_z=-r_z\sharp_h r_z$. 
We also deduce from $\gamma_1=\partial_\xi H_h[\xi]$ that $\partial_\xi(z-H_h)^{-1}=(z-H_h)^{-1} \partial_\xi H_h (z-H_h)^{-1}=(z-H_h)^{-1}\gamma_1 (z-H_h)^{-1}$ that at the level of symbols, $\partial_\xi r_z = r_h\sharp_h \gamma_1 r_z$.

We thus find
\[
  \gamma_1 \partial_\lambda r_z = -\gamma_1 r_z \sharp_h r_z  + r_z\sharp_h \gamma_1 r_z -\partial_\xi r_z.
\]
Let us introduce the rectangle $\R=\prod_{l=1}^3(-M_l,M_l)$ where we identify the variables $(1,2,3)$ with $(\xi,\zeta,y)$. Integrals such as $I_\delta$ above are seen as the limit as  $\R\to\infty$ with the meaning that $M_l\to\infty$ for each $l=1,2,3$. 
We are therefore interested in the quantity for $\omega>0$
\[
  \Phi_\R(z) := \frac1h  \dint_{\R} {\rm tr} \gamma_1 \partial_\lambda r_z d\R = \frac1h  \dint_{\R} {\rm tr} [-\gamma_1 r_z \sharp_h r_z  + r_z\sharp_h \gamma_1 r_z -\partial_\xi r_z] d\R.
\]
We want to show that the above trace is asymptotically in divergence form and thus an integral over $\partial\R$.  For two symbols $a$ and $b$, we find \cite{dimassi1999spectral,zworski2012semiclassical}
\[
  a\sharp_h b = ab + \frac{h}{2i}\{a,b\} + O(h^2)
\]
with $\{a,b\}=\partial_\zeta a\partial_yb-\partial_y a \partial_\zeta b$. By cyclicity of the trace, $\tr ab=\tr ba$ and $\tr\{a,b\}=-\tr\{b,a\}$ so that $\tr a\sharp_hb -\tr b\sharp_h a =- ih \tr\{a,b\} + O(h^2)$. Here and below, all terms $O(h^k)$ are so in the sense of $S^0(1)$ symbols \cite{dimassi1999spectral}.  This implies
\[
   -\gamma_1 r_z \sharp_h r_z  + r_z\sharp_h \gamma_1 r_z = ih {\rm tr} \{\gamma_1r_z,r_z\} +O(h^2).
\]
The integral $I_\delta$ involves a factor $h^{-1}$ so that only $i{\rm tr} \{\gamma_1r_z,r_z\}$ survives in the limit $h\to0$. It thus remains to integrate ${\rm tr} [ih  \{\gamma_1r_z,r_z\} -\partial_\xi r_z]$ over $\R$.

We now verify that $2\{a,b\}=\partial_\zeta(a\partial_yb-\partial_ya b)-\partial_y(a\partial_\zeta b-\partial_\zeta ab)$ in divergence form in these variables.  We thus find that  
\[
   {\rm tr} \{\gamma_1r_z,r_z\}= \partial_y \tilde a_y  -\partial_\zeta \tilde a_\zeta, \quad \tilde a_\tau=\frac12[\gamma_1 r_z,\partial_\tau r_z], \quad \tau=y,\zeta.
\]
This means that the integral of ${\rm tr} [ih  \{\gamma_1r_z,r_z\} -\partial_\xi r_z]$ over $\R$ may be written as
\[
   \Phi_\R(z) = \frac 1h {\rm tr} \int_{\partial \R} [ ih (\tilde a_y \nu_y -\tilde  a_\zeta \nu_\zeta) + r_z\nu_\xi] d\Sigma + O(h),
\]
with $\nu_\tau$ the $\tau$th component of the normal unit vector to $\R$ at the boundary and $d\Sigma$ the surface measure there. Provided that $M_l$ are large enough, $\sigma_z^{-1}$ is then defined by assumption (h2) for all $\lambda$ sufficiently close to $\alpha$ by continuity. That only values of $\lambda$ close to $\alpha$ are in $Z$ is ensured by possibly shrinking the support of $\varphi'$ and replacing it by $\varphi_\rho'(\lambda) = \frac1\rho \varphi'(\frac1\rho(\lambda-\alpha)+\alpha)$. We know from Proposition \ref{prop:stabsigma} that $\sigma_I$ is independent of $0<\rho\leq 1$. 

The inverse $\sigma_z^{-1}$ now exists uniformly in $\omega$, even $|\omega|\leq\delta$, still by assumption. This allows us to approximate $r_z$ using ${\rm Op}^w_h(r_z) {\rm Op}^w_h(\sigma_z) = I$ so that  $I = r_z \sharp_h \sigma_z$. By  semiclassical calculus of Weyl operators this  implies that 
\[
   I = r_z \sigma_z + \frac{h}{2i}\{r_z,\sigma_z\} + O(h^2).
\]
From this, we deduce that in the $S^0(1)$ sense
\begin{equation}\label{eq:rzapp}
  r_z = \sigma_z^{-1} + \frac{ih}{2} \{ r_z,\sigma_z\} \sigma_z^{-1} + O(h^2) = \sigma_z^{-1} + \frac{ih}{2} \{ \sigma^{-1}_z,\sigma_z\} \sigma_z^{-1} + O(h^2).
\end{equation}
To leading order, we may therefore replace $r_z$ by $\sigma_z^{-1}$ in the terms $\tilde a_\tau$ above, defining the corresponding terms $a_\tau$ and obtain
\[
  \Phi_\R(z) =  {\rm tr} \int_{\partial \R} [ i (a_y \nu_y - a_\zeta \nu_\zeta) +  \frac{i}{2} \{ \sigma^{-1}_z,\sigma_z\} \sigma_z^{-1}\nu_\xi + \frac1h \sigma_z^{-1}\nu_\xi] d\Sigma + O(h).
\]
We now write, using that $\gamma_1=-\partial_\xi \sigma_z$ is a constant matrix,
\[
  2\tr \{\sigma_z^{-1},\gamma_1 \sigma_z^{-1}\} = \tr [\partial_\zeta ([\partial_y\sigma_z^{-1},\partial_\xi\sigma_z]\sigma_z^{-1}) - \partial_y ([\partial_\zeta\sigma_z^{-1},\partial_\xi\sigma_z] \sigma_z^{-1})].
\]
The last term equals $\partial_y ([\partial_\xi\sigma_z^{-1},\partial_\zeta\sigma_z] \sigma_z^{-1})$. We also use the identity $\partial_i\sigma_z^{-1}=-\sigma_z^{-1}\partial_i\sigma_z \sigma_z^{-1}$ to find eventually that
\[
\Phi_\R(z) =  \frac i2 {\rm tr} \int_{\partial \R} [ \eps_{ijk} \sigma_z^{-1} \partial_i\sigma_z
 \sigma_z^{-1} \partial_j\sigma_z \sigma_z^{-1}\nu_k +  \frac1h \sigma_z^{-1}\nu_\xi] d\Sigma + O(h).
\]
We thus obtain that
\[\begin{array}{rcl}
  \Psi_\R(z) &:=& \dfrac 1h {\rm tr} \dint_\R \gamma_1 r_z d\R = -i \dint_0^\infty \Phi_\R(z+it)dt 
 \\ &=& \dfrac12 \dint_{\partial\R \times[0,\infty)}  [ \eps_{ijk} \sigma_{z+it}^{-1} \partial_i\sigma_{z+it}
 \sigma_{z+it}^{-1} \partial_j\sigma_{z+it} \sigma_{z+it}^{-1}\nu_k +  \frac1h \sigma_{z+it}^{-1}\nu_\xi] d\Sigma dt.
\end{array}
\]
We have $I_\delta=\lim_{\R\to\infty} I_{\delta,\R}$ with 
\[
  I_{\delta,\R} =\frac{-1}{2\pi^2} \dint_{Z_\delta} \bar\partial \widetilde{\varphi'}(z) \Psi_\R(z) dz.
\]
This integral is close to an integer as $\delta\to0$ and $\R\to\infty$ since $2\pi\sigma_I$ is known to be an integer. The term proportional to $h^{-1}$ in $\Psi_R(z)$ must therefore (exactly) vanish. We still call $\Psi_\R(z)$ the term without his component.
Since $\Psi_\R(z)$ is a smooth function in $z$ close to $\alpha$ even when $\omega$ is small, then up to an error of order $\eps$,
\[
 I_{\delta,\R} =\frac{-1}{2\pi^2} \dint_{Z} \bar\partial \widetilde{\varphi'}(z) \Psi_\R(z) dz + o(1).
\]
Since all terms involved (from the beginning) are analytic in $z$, then so is $\Psi_R(z)$ away from the real axis $\omega=0$. Let $Z=\overline{Z_+\cup Z_-}$, where $Z_\pm=Z\cap\{\pm\omega>0\}$.
Using $\bpar  \widetilde{\varphi'}(z) \Psi_\R(z) = \bpar( \widetilde{\varphi'}(z) \Psi_\R(z))$, we have by the Stokes formula
\[
   \dint_{Z_\pm} \bar\partial \widetilde{\varphi'}(z) \Psi_\R(z) d^2z
   = \frac{-i}2\dint_{\partial Z_{\pm}} \widetilde{\varphi'}(z) \Psi_\R(z) dz,
\]
where $dz=\pm d\lambda$ on the part of $\partial Z_{\pm}$ on the real axis, $\widetilde{\varphi'}(z)=\varphi'(\lambda)$ while $\Psi_\R(z)$ is defined as the limit from $\omega>0$ while $\Psi_\R(\bar z)=\bar \Psi_R(z)$ for $\omega<0$. Overall, we therefore find 
\[
  o(1)+ I_\R = \frac{i}{4\pi^2}  \dint_\Rm \varphi'(\lambda) [\Psi_R(\lambda+i0)-\bar\Psi_R(\lambda-i0)] d\lambda.
\]
Each of the $\Psi_R$ terms  is an integral over a half line in $\omega$ while the difference involves a line integral in $\omega$.

As mentioned earlier, this term is independent of $\varphi_\rho'(\lambda) = \frac1\rho \varphi'(\frac1\rho(\lambda-\alpha)+\alpha)$. As $\rho$ to $0$, and using the smoothness of $\lambda\mapsto \Psi_R(\lambda+i0)$, we evaluate the above at $\lambda=\alpha$ with $\int _\Rm\varphi'_\rho(\lambda)d\lambda=1$ to deduce in the limit $\R\to\infty$ the equality of integers:
\[
  2\pi \sigma_I  = \lim_{\R\to\infty} \frac{i}{8\pi^2} \dint_{\partial\R \times \Rm} {\rm tr} [ \eps_{ijk} \sigma_\alpha^{-1} \partial_i\sigma_\alpha \sigma_\alpha^{-1} \partial_j\sigma_\alpha \sigma_\alpha^{-1}\nu_k ] d\Sigma  d\omega.
\]
This concludes the proof of the proposition. 
\end{proof}

We now deduce several corollaries from the preceding result. The most natural result is the following relation between the interface conductivity and the Fedosov-H\"ormander formula for the index of the operator $i\omega-H$.  
\begin{corollary}[Spectral flow Index]
\label{cor:sfi}
  Let $\Sigma$ be a surface in the variables $(\xi,\zeta,y,\omega)$  such that $\sigma_\alpha^{-1}$ is defined on and outside of the surface. Then
\begin{equation}\label{eq:sfi}
  2\pi \sigma_I = \dfrac{1}{24\pi^2} \dint_\Sigma {\rm tr}  (\sigma_\alpha^{-1}d\sigma_\alpha)^{\wedge 3}.
\end{equation}
\end{corollary}
Here and below, the orientation of $\Rm^4$ is $d\xi\wedge d\zeta\wedge dy\wedge d\omega=d\omega\wedge d\xi\wedge dy \wedge d\zeta>0$.
\begin{corollary}[Bulk Interface Correspondence]
\label{cor:bic}
 We find that 
\begin{equation}\label{eq:bic}
    2\pi \sigma_I = \lim_{y\to\infty} \dfrac{i}{8\pi^2} \dint_{\Rm^3} \tr [\sigma_\alpha^{-1}\partial_\xi\sigma_\alpha,\sigma_\alpha^{-1}\partial_\zeta\sigma_\alpha] \sigma_\alpha^{-1}  \big|_{-y}^y d\xi d\zeta d\omega = -W_\alpha = \sum_{h_i<\alpha}c_i.
\end{equation}
\end{corollary}
This is the bulk-interface correspondence, stating that the interface conductivity equals the bulk-difference invariant. We use the notation $X|_{-y}^y=X(y)-X(-y)$. We also have the relation to the spectral asymmetry:
\begin{corollary}[Spectral Asymmetry]\label{cor:sa}
Let us assume that the $3-$form $\tr (\sigma_\alpha^{-1}d\sigma_\alpha)^{\wedge 3}$ integrated over bounded domain in $(\omega,\xi,\zeta)$ at fixed values of $y$ converges to $0$ as $|y|\to\infty$. 
Then we find that:
\begin{equation}\label{eq:sa}
    2\pi \sigma_I = \lim_{\xi\to\infty} \dfrac{i}{8\pi^2} \dint_{\Rm^3} \tr [\sigma_\alpha^{-1}\partial_\zeta\sigma_\alpha,\sigma_\alpha^{-1}\partial_y\sigma_\alpha] \sigma_\alpha^{-1} \big|_{-\xi}^\xi dyd\zeta d\omega.
\end{equation}
\end{corollary}
This is the three-dimensional winding number of  the Green's function used in \cite{essin2011bulk,volovik2009universe}.
\begin{proof}
We first note that $i=\partial_\omega\sigma_\alpha$ with $\sigma_\alpha=\sigma_\alpha(\omega,y,\xi,\zeta)$. The result of Proposition \ref{prop:sa} states that $\sigma_I$ is up to a constant the limit as $\R\to\infty$ of the integral of $\tr (\sigma_\alpha^{-1} d\sigma_\alpha)^3$ over the cylinder $\Rm\times \R$. A direct calculation (as in the proof of the index theorem in \cite[Chapter 19.3]{H-III-SP-94}) shows that 
\[
  \tr (\sigma_\alpha^{-1} d\sigma_\alpha)^4 =0,
\]
which implies that $\tr (\sigma_\alpha^{-1} d\sigma_\alpha)^3$ is a closed form. 

Note that the latter form equals $3$ times the form written using commutators as in the proposition; this explains the factor $24=8\times3$ in the first corollary. For $M_i$ sufficiently large, let $\Omega$ be the volume with disconnected boundary terms $\Sigma$ and $\partial \R_4$, where $\R_4=\prod_{l=1}^4(-M_l,M_l)$ this time also including the fourth variable identified with $\omega$. Then the Stokes theorem implies that the integral of $\tr (\sigma_\alpha^{-1} d\sigma_\alpha)^3$ over $\Sigma$ and that over $\partial \R_4$ (with the same orientation) are equal. 

To prove Corollary \ref{cor:bic}, we need to approximate the integral over the cylinder $\Rm_\omega\times\partial \R_3$ where $\R_3=\prod_{l=1}^3(-M_l,M_l)$ by an integral over the boundary of a four-dimensional rectangle $\R_4=\prod_{l=1}^4(-M_l,M_l)$. This means that, as $M_4=M_\omega\to\infty$, we want the integral of $(\sigma_\alpha^{-1} d_{\xi,\zeta,y}\sigma_\alpha)^{\wedge 3}$ over the rectangle $\R_3$ for a fixed value $\omega_0$ of $\omega$ to be small as $\omega_0\to\infty$. From (h2), $\sigma_\alpha^{-1}d\sigma_\alpha$ is bounded (uniformly as a matrix-valued one form) and $\partial_\xi\sigma_\alpha=-\gamma_1$ bounded while $\sigma_\alpha^{-1}$ goes to $0$ (uniformly in $(y,\xi,\zeta)$ as a $N\times N$ matrix) as $\omega_0\to\infty$.  Passing to the limit in $\R\to\infty$ gives the spectral flow index in Corollary \ref{cor:bic}. 

Note that the above derivation shows that the limit as $\R\to\infty$ is replaced by the integral of the trace of the matrix-valued three-form over any closed (three-dimensional) surface $\Sigma$.

 The other corollaries are obtained similarly by appropriately deforming $\Sigma$ to two surfaces since $\tr (\sigma_\alpha^{-1} d\sigma_\alpha)^3$ is a closed form. To prove Corollary \ref{cor:sa}, which aims to compute the invariant by integrals at fixed values of $\xi$,  we need to ensure that the integral of $(\sigma_\alpha^{-1} d_{\zeta,y,\omega}\sigma_\alpha)^{\wedge 3}$ over any (lateral) rectangle (involving an integration in $\xi$) goes to $0$ as $(\omega,y,\zeta)\to\infty$. To prove Corollary \ref{cor:sfi}, which aims to compute the invariant by integrals at fixed values of $y$, we need to ensure that the integral of $(\sigma_\alpha^{-1} d_{\xi,\zeta,\omega}\sigma_\alpha)^{\wedge 3}$ over any (lateral) rectangle (involving an integration in $y$) goes to $0$ as $(\omega,\xi,\zeta)\to\infty$. 

We already know that the integrals over (lateral) rectangles with constant values of $\omega$ converge to $0$ as $|\omega|\to\infty$. The same argument shows the same result for integrals over (lateral) rectangles at constant values of $|\xi|\to\infty$. By hypothesis (h2) on the behavior of $\sigma_\alpha$ in  $\zeta$, we also deduce the result for  integrals over (lateral) rectangles at constant values of $|\zeta|\to\infty$. This proves Corollary \ref{cor:bic}. To prove Corollary \ref{cor:sa}, it therefore remains to obtain that the integral over a fixed rectangle in the variables $(\omega,\xi,\zeta)$ goes to $0$ as $|y|\to\infty$. This is in fact incorrect in general and this is why it was added as an assumption in the corollary.
\end{proof}

The results \eqref{eq:sfi} and \eqref{eq:bic} conclude the proof of Theorem \ref{thm:bic} except for the last equality involving the index of the Fredholm operator $F$. Under conditions (H1') or (H2'), we proved in proposition \ref{prop:fredholm} that $F$ was indeed Fredholm. That its index is given by the formula in the theorem can be found or proved as in \cite[Theorem 19.3.1]{H-III-SP-94}. This concludes the proof of  Theorem \ref{thm:bic}. $\square$

\begin{proposition} 
 For the  systems \eqref{eq:s22} and \eqref{eq:s33}, the assumptions in Corollary \ref{cor:sa} hold when $m(y)$ and $f(y)$, respectively, converge to $\infty$ at $\infty$.
\end{proposition}
\begin{proof}
  In both cases, at fixed values of $y$ the integral of $\tr[\sigma_\alpha^{-1}\partial_\xi\sigma_\alpha,\sigma_\alpha^{-1}\partial_\zeta\sigma_\alpha]\sigma_\alpha^{-1}$ over a bounded domain at fixed values of $y$ involves an integrand that goes to $0$ when $|\mu(y)|$ converges to $\infty$.
\end{proof}
The above result states that while Corollaries \ref{cor:sfi} and \ref{cor:bic} always hold, the regularized spectral asymmetry \eqref{eq:sa} or \eqref{eq:heursf} in \cite{essin2011bulk,gurarie2011single} holds only when the range of the domain wall is the whole of $\Rm$ (in the topologically non-trivial case). A calculation for system \eqref{eq:s22} shows that \eqref{eq:sa} does not hold when $m(y)$ takes values between $-m_0<0$ and $m_0$, say. Indeed, the bulk-difference invariant $c_+$ is given by \cite{B-BulkInterface-2018}
\[
  \dfrac{2}{4\pi} \dint_{\Rm^2} \dfrac{m}{(\xi^2+\zeta^2+m^2)^{\frac32}} d\xi d\zeta=\sgn{m},
\]
whereas the spectral asymmetry `invariant' would be given by
\[
  \dfrac{-2}{4\pi} \dint_{\Rm^2} \dfrac{m'(y) \xi}{(\xi^2+\zeta^2+m^2)^{\frac32}} d\zeta dy =  \dfrac{-2}{4\pi} \dint_{\Rm\times[-m_0,m_0]} \dfrac{\sgn{m_0}\xi}{(\xi^2+\zeta^2+m^2)^{\frac32}} d\zeta dm.
\]
Calculated for $\xi>0$ large enough, it is equals to $c_+$ only in the limit $m_0\to+\infty$ and in fact converges to $0$ as $\xi\to\infty$. 

\medskip

For \eqref{eq:s33}, we thus find that for smooth and sufficiently slowly varying profiles of the Coriolis force $f(y)$, the bulk-interface correspondence \eqref{eq:bic} always holds. This was one of the main motivations for the theory presented in this section. 

The notion of index related to a spectral flow presented in Appendix \ref{sec:bic} finds a mathematical explanation in Corollary \ref{cor:sfi}. The spectral flow, or equivalently the spectral asymmetry, is indeed represented as the index of an extended operator taking the form $ix-H$, where $\omega$ and the original variable $x$ may be identified. Now, $-ix+\gamma_1 D_x$ may be written in the Fourier domain as $\partial_\xi + \gamma_1\xi$, which is the augmented operator introduced in Appendix \ref{sec:bic} below.

%
\section{Conclusions}
\label{sec:conclu}

This paper proposes two methods to compute the topologically quantized interface conductivity $2\pi\sigma_I$ of Hamiltonians of the form $H[\mu(y)]+V$ where $\mu(y)\in\fS[\mu_-,\mu_+]$ implements a domain wall and $V$ models a large class of (local) perturbations.

The first method presented in section \ref{sec:ii} is a direct application of the heuristics of spectral flows. It requires a diagonalization of the domain wall problem and shows that the conductivity is given by the sum of the winding numbers of the branches of absolutely continuous spectrum that cross a spectral region of interest. The mathematical backbone of the derivation borrows from the machinery of non-commutative geometry to construct a Fredholm operator of the form $PU(H)P$ whose index is given by the Fedosov formula $-{\rm Tr}[P,U]U^*$ and as a topological object is immune to large classes of perturbations $V$.  This immunity is the main reason for the practical interest of such theories and gives a quantitative meaning to the perceived topological protection of transport phenomena observed numerically or experimentally in several areas of materials science and geophysical fluid flows.

 The main drawback of the method is that it requires in practice a sufficiently explicit expression of the spectral decomposition of $H[\mu(y)]$. Its main advantage is that it applies even when $\mu(y)$ displays discontinuities. It allows one to justify the surprising result for the system \eqref{eq:s33} that the interface conductivity may take different quantized values for different profiles in the same topological class of domain walls $\fS[\mu_-,\mu_+]$. 

\medskip

The second method is based on computing a bulk invariant involving only the constant-coefficient Hamiltonians $H[\mu_\pm]$ and relating it to the interface conductivity $2\pi\sigma_I$. We introduced a notion of bulk-difference invariant in section \ref{sec:bulk} that appropriately combines the projectors associated to both bulk Hamiltonians $H[\mu_\pm]$ and is computed either by means of Chern numbers on the sphere or three-dimensional winding numbers of associated Green's functions. These invariants are typically easier to compute than the ones introduced in section \ref{sec:ii}; see \cite{BM-FTI-21} for an application of Floquet topological insulators, where the computations of these bulk invariants is possible while the spectral decomposition of the domain walls seems intractable. 

Once these bulk invariants are computed, they are related to the interface conductivity by a general principle called here a bulk-interface correspondence. While the latter does not always hold as we just mentioned, we showed in section \ref{sec:BI} that it applied to a large class of differential operators when their coefficients are sufficiently smooth, and for the system \eqref{eq:s33} sufficiently slowly varying. 

The correspondence is proved by relating both indices to a third one given as the index of a Fredholm differential operator constructed from $H[\mu(y)]$ by adding a second domain wall $i \nu(x)$ (and hence enforcing localization) in the direction of propagation along the interface. We first introduced some sufficient (pseudo-)differential conditions ensuring that the conductivity was quantized. We then related the conductivity to an integral involving only the symbol of the Hamiltonian $H[\mu(y)]$ and identified this integral with the Fedosov-H\"ormander formula for the index of the differential operator.

Our class of operators was restricted to the form \eqref{eq:Hh} that includes generalizations of models such as \eqref{eq:s22} and \eqref{eq:s33} or the Floquet model in \cite{BM-FTI-21}. The method extends to a larger class of topologically nontrivial differential operators such as those appearing in superconductor theory \cite{volovik2009universe}; see \cite{QB-NUMTI-2021}.

\section*{Acknowledgment}
This research was partially supported by the Office of Naval Research, Grant N00014-17-1-2096 and by the National Science Foundation, Grants DMS-1908736 and EFMA-1641100.

\appendix

%

%
\section{Index of interface Hamiltonian}
\label{sec:appindex}
%

The proof of Theorem \ref{thm:intunpert} is similar to those derived in \cite{B-BulkInterface-2018} to compute the index of interface Hamiltonians. We highlight the main differences. It is shown in \cite{B-BulkInterface-2018} that under the assumptions on the spectral decomposition of $U(H)$, then $PU(H)P_{|{\rm Ran}P}$ is a Fredholm operator and
\[
 -{\rm Index}(PUP) = {\rm Tr} [P,U]U^* = {\rm Tr} [P,W](I+W^*) = {\rm Tr} [P,W]W^* + {\rm Tr}[P,W]
\]
since $[P,W]$ is trace-class. Moreover, as in \cite{B-BulkInterface-2018,brislawn1988kernels}, the above traces and given by the integrals along the diagonals of their Schwartz kernels. The trace of $[P,W]$ is seen to vanish since $(p(x)-p(y))_{|y=x}=0$. The Schwartz kernel of $W$ is given by
\[
  w(x-x';y,y') = \dint_{\Rm} \dsum_j W(E_j(\xi)) \psi_j(y,\xi)\psi_j^*(y',\xi) \dfrac{e^{i(x-x')\xi}}{2\pi} d\xi.
\]
The kernel of $W^*$ is given by $w^*(x'-x;y',y)$ with $w^*_{ij}=\bar w_{ji}$. The kernel of $[P,W]W^*$ is thus given by
\[
  t(x,x';y,y') = \dint_{\Rm^{d+1}} (\chi(x)-\chi(x")) w(x-x";y,y")w^*(x'-x";y',y")dx"dy",
\]
where $\chi(x)$ is the kernel of the projection $P$. Therefore, $T:={\rm Tr} [P,W]W^*$ is computed by
$T = \int_{\Rm^{d+1}} t(x,x;y,y) dxdy$.
Using the change of variables $(x,x")\to(z,x")=(x-x",x")$ with $dxdx"=dzdx"$, and computing $\int_\Rm (\chi(x"+z)-\chi(x"))dx"=z$, we obtain
\[
  T ={\rm tr} \dint_{\Rm^{2d+1}} z  w(z;y,y') w^*(z;y,y') dz dy dy'.
\]
Using the Fourier transform from $z$ to $\xi$ yields by Parseval
\[
  T  = \dfrac{-{\rm tr}}{2\pi i } \dint_{\Rm^{2d+1}} \partial_\xi \hat w(\xi;y,y') \hat w^*(\xi;y,y') d\xi dy dy',
\]
where $\hat w(\xi;\cdot)$ is the component-wise Fourier transform of $w(x;\cdot)$ given by
\[
  \hat w(\xi;y,y') = \dsum_j W(E_j(\xi)) \psi_j(y,\xi) \psi^*_j(y',\xi).
\]
The derivative $\partial_\xi$ applies to $W\circ E_j$ and to $\psi_j(y,\xi) \psi^*_j(y',\xi)$. At a fixed $\xi$, consider the latter contribution, which is given by
\[
  \tau(\xi) := \dint \dsum_{j,k} {\rm tr} \partial_\xi[\psi_j(y,\xi)\psi_j^*(y',\xi)] \psi_k(y',\xi) \psi_k^*(y,\xi) dy dy'.
\]
We show that $\tau(\xi)=0$.  Indeed, we distribute $\partial_\xi$ over the product, exchange $y$ and $y'$ in the second contribution to get (dropping the $\xi$-dependence to simplify notation)
\[
  {\rm tr} \dint\dsum_{j,k} [\partial_\xi\psi_j(y)\psi_j^*(y') \psi_k(y')\psi_k^*(y) + \psi_j(y')\partial_\xi\psi_j^*(y).\psi_k(y)\psi_k^*(y')] dy dy'.
\]
Applying traces to these products of rank-one matrices yields (with $\bar \psi$ as a column vector)
\[
  \tau(\xi) = \dint \dsum_{j,k} [\partial_\xi \psi_j(y)\cdot \bar \psi_k(y) \bar \psi_j(y')\cdot \psi_k(y')  + \partial_\xi \bar\psi_j(y)\cdot \psi_k(y) \psi_j(y')\cdot \bar \psi_k(y') ] dy dy'.
\]
By orthogonality of the eigenvectors, only the terms $j=k$ survive the integration and then
$\tau(\xi) = \dint \dsum_j \partial_\xi |\psi_j(y)|^2 dy = \dsum_j \partial_\xi \dint |\psi_j(y)|^2 dy=0$. 
As a consequence,
\[
  T = \frac{-\rm tr}{2\pi i} \dsum_{j,k} \dint \partial_\xi W\circ E_j (\xi)  \psi_j(y,\xi) \psi_j^*(y',\xi)W^*\circ E_j (\xi) \psi_k(y',\xi) \psi_k^*(y,\xi) dy dy' d\xi.
\]
Taking traces again and integrating in $y$ and $y'$ yields
\[
 T = \frac{-1}{2\pi i} \dsum_j  \dint \partial_\xi W\circ E_j(\xi) W^*\circ E_j (\xi)  d\xi = -\dsum_j  {\mathcal W}_1(W\circ E_j).
\]
Here, ${\mathcal W}_1(f)$ is the winding number of a (compactly supported) function $f$. Since the integral of $W'$ vanishes, the winding number of $W\circ E_j$ and that of $U\circ E_j$ are the same. This concludes the proof of Theorem \ref{thm:intunpert}.

\medskip

%
%

%
%
\section{Calculations of conductivities in systems \eqref{eq:s22} and \eqref{eq:s33}}
\label{sec:appli}
We show that the corresponding operators are unbounded self-adjoint operators on $L^2(\Rm^2)\otimes\Cm^n$ and compute the interface conductivity in some tractable cases.

\medskip
\noindent{\bf The $2\times2$ system \eqref{eq:s22}}. We briefly recall and slightly extend results computed in \cite{B-BulkInterface-2018} for this system so that they can be contrasted with the more challenging $3\times3$ system. The domain of definition of $H$ making it self-adjoint is given by ${\mathcal D}(H):=(H+\epsilon i)^{-1}{\mathcal H}$ with ${\mathcal H}= L^2(\Rm^2)\otimes\Cm^2$. Here $\epsilon=\pm1$. Following \cite[Chapters 4\&8]{dimassi1999spectral}, we verify that the above domain is independent of $\epsilon$. Indeed the source problem, with $\epsilon=1$, may be recast as
\[
  (m+i)\psi_1 + (D_x-iD_y)\psi_2 = s_1,\quad (D_x+iD_y)\psi_1 + (i-m)\psi_2 = s_2
\]
which after elimination of one variable yields, for instance,
\[
   \Big( (m+i) + (D_x-iD_y) \dfrac{-1}{i-m} (D_x+iD_y) \Big) \psi_1 =  s_1-(D_x-iD_y)\frac{1}{i-m}s_2.
\]
This is an elliptic problem that implicitly defines a $m(y)-$dependent domain of definition. The same domain of definition applies to $\psi_2$. 

Now, for any bounded (multiplication or not) operator $V$, the problem $(H+V+\epsilon i)\psi=0$ implies that $\psi=0$ since $((H+V)\psi,\psi)$ is real and $(\epsilon i \psi,\psi)$ is purely imaginary. This shows that for such perturbations, $H_V=H+V$ is self-adjoint and the spectral calculus applies.

\medskip 

\noindent{\bf Spectral decomposition of system \eqref{eq:s22}.}
 Let us first consider $m(y)$ either of the form $m(y)=\lambda y$ or $m(y)$ a switch function in $\fS[m_-,m_+]$ with $m_-m_+<0$. Detailed calculations are in \cite{B-EdgeStates-2018,B-BulkInterface-2018}. We summarize the main relevant results.

Let us consider $H[\xi]$, the partial Fourier transform from $x$ to $\xi$, which we represent \cite{B-EdgeStates-2018,B-BulkInterface-2018} as $D_y\sigma_1+m(y)\sigma_2+\xi\sigma_3$ after a global unitary transform permuting $\sigma_{1,2,3}$ to $\sigma_{3,1,2}$. In case (i), $H[\xi]$ has compact resolvent, while in case (ii), its restriction to the interval $(-m_0,m_0)$ with $m_0={\rm min}(|m_+|,|m_-|)$, is also compact. As $\xi$ varies, this provides branches of absolutely continuous spectrum for the operator $H[m(y)]$.

In either case, we observe that solving $(H[\xi]-E)\psi=0$, i.e., finding the edge states,  amounts to solving
\[
  (\xi-E)\psi_1-i\fa \psi_2 =0,\quad i\fa^*\psi_1-(\xi+E)\psi_2 =0,\qquad \fa = \partial_y+m(y),\quad \fa^*=-\partial_y+m(y),
\]
which leads to trivial branches of spectrum corresponding to the shared (strictly) positive eigenvalues of $\fa\fa^*$ and $\fa^*\fa$, and one special branch of spectrum corresponding to the nontrivial kernel of $\fa$ when $m_-<0<m_+$ (the case we now consider). The branch is given by $\fa\psi_2=0$, $\psi_1=0$, and $E(\xi)=-\xi$, with group velocity $E'(\xi)=-1$ corresponding to a mode moving towards negative values along the $x$ axis. In other words, the operator $H[m(y)]$ admits the branch of continuous spectrum (parametrized by $E$)
$
  e^{-M(y)} e^{-iE x} \begin{pmatrix} 0\\1\end{pmatrix}
$
at least for values $|E|<m_0$ in case (ii), where $M$ is the antiderivative of $m$ with $M(0)=0$. 

The above calculations imply that $2\pi\sigma_I=-\sgn{m_+}$ when $m(y)\in\fS[\mu_-,\mu_+]$ and $m_+m_-<0$ (with $\sigma_I=0$ when $m_+m_->0$) while $2\pi\sigma_I=-\sgn{\lambda}$ when $m(y)=\lambda y$ (and would vanish for $m(y)=\lambda|y|$).

\medskip
\noindent
{\bf The $3\times3$ system \eqref{eq:s33}.} We start by analyzing its Hermitian structure.
We consider the $3\times3$ system with perturbation given by $H_V=H+V$. We define the domain of such an operator as ${\mathcal D}(H)=(H+\epsilon i)^{-1}{\mathcal H}\subset{\mathcal H}$ with ${\mathcal H}=L^2(\Rm^2)\otimes\Cm^3$ with $\epsilon=\pm1$ and then show that $H+V+\epsilon i$ is still surjective from ${\mathcal D}(H)$ to ${\mathcal H}$, or equivalently that the kernel of $H+V+\epsilon i$ on that domain is trivial.
The problem $(H+i)\psi=S$, with $\epsilon=1$ to simplify may be recast as
\[
 \left(\begin{matrix} i+D_x\sigma_1 & B\\ B^* & i  \end{matrix}\right) 
 \left(\begin{matrix} w\\v  \end{matrix}\right)  =  S = \left(\begin{matrix} S_w\\S_v  \end{matrix}\right) 
\]
where we have defined $w=(\eta,u)^t$ and $B=(D_y,if(y))^t$. We find
$  (i+D_x\sigma_1)^{-1} = -(1+D_x^2)^{-1} (i-D_x\sigma_1)$
so that the above system is equivalent to
\[
  (B^*(1+D_x^2)^{-1} (i-D_x\sigma_1)B + i) v = S_v +(1+D_x^2)^{-1} (i-D_x\sigma_1) S_w,
\]
along with $w=(1+D_x^2)^{-1} (S_w+(i-D_x\sigma_1)Bv)$. The equation for $v$ is therefore
$\big(B^*(i-D_x\sigma_1)B + i (1+D_x^2)\big) v = (1+D_x^2)S_v + (i-D_x\sigma_1) S_w$
or equivalently
\[
  \big(i (D_y^2+f^2) +f'(y)D_x + i (1+D_x^2)\big) v = (1+D_x^2)S_v + (i-D_x\sigma_1) S_w
\]
This is an elliptic problem with a unique solution, which provides a domain of definition for $v$ given by $(D_x^2+D_y^2+\aver{f}^2)^{-1}(1+D_x^2) L^2(\Rm^2)$. This propagates to a domain of definition for $(\eta,u)$ coming from $w=(1+D_x^2)^{-1} (S_w+(i-D_x\sigma_1)Bv)$. We observe that the spaces are independent of $\epsilon=\pm1$. 

For any bounded Hermitian perturbation $V$, the same spaces provide that $(H+V+\epsilon i)\psi=0$ implies that $\psi=0$. Indeed on that domain $((H+V)\psi,\psi)$ is real while $(\psi,\epsilon i\psi)$ is purely imaginary. This shows that $H_V=H+V$ is a self-adjoint operator on ${\mathcal D}(H)$ and the spectral calculus applies.

\medskip
\noindent
{\bf Spectral decomposition of system \eqref{eq:s33}.}
Consider now the Hamiltonian $H[f(y)]$, which is invariant by translation along the $x-$axis. One difficulty not present in the $2\times2$ system is the presence essential spectrum at energy $0$. The bulk-interface correspondence predicts a number of protected interface modes equal to $2$, which is `often' the case, as we demonstrate in section \ref{sec:BI}, but not always. We now compute the winding number of branches of continuous spectrum for several profiles $f(y)$. Similar calculations have been worked out in \cite{delplace2017topological,souslov2019topological,tauber2019bulk,tauber2019anomalous} without explicitly looking to identify spectral branches.

We are looking at continuous branches of spectrum of $H[f(y)]$, i.e., discrete spectrum of $H[\xi]$ that is continuous in $\xi$. This requires solving for $H[\xi]\psi=E(\xi)\psi$ for $\psi=(\eta,u,v)^t$, which is given by the following system:
\begin{equation}\label{eq:33}
  \xi u + Dv =E\eta,\quad \xi\eta+ifv=Eu,\quad D\eta-if u=Ev,
\end{equation}
with $D=\frac1i\partial_y$.  We first make general remarks on such systems and then consider specific examples that can be solved explicitly. For concreteness, we consider the setting where $f'(y)\geq0$. We first look for solutions with $v=0$. We find from the other equations that $|E|=|\xi|$ and $u=\pm\eta$ when $E\not=0$ with $\partial_y \eta + fu=0$. This shows that $\eta=u$ provides a solution $u=e^{-F(y)}$ with $F'(y)=f(y)$ and $F(0)=0$, say, which is unique up to normalization. Note that $\fa u=0$ with $\fa=\partial_y+f$ an operator with non-trivial kernel and trivial co-kernel (and hence a Fredholm operator in appropriate topologies with index equal to $1$).  Note that $E=E(\xi)=\xi$ provides a branch of continuous spectrum. Moreover, $E'(\xi)=1$ corresponds to waves propagating toward positive values of $x$; these are the eastward propagating Kelvin waves in the geophysical application.

Still with $v=0$, it remains to look at the case $E=\xi=0$ where we find that $D\eta=ifu$ as the only constraint. There is an infinite number of solutions and $E=\xi=0$ corresponds to essential spectrum that needs to be avoided. 
We thus assume $E\not=0$.

Let us now look for solutions with $v\not=0$. Eliminating $\eta$, assuming $\xi\not=0$, yields
\[
 (\xi D+ifE)v=(E^2-\xi^2)u,\quad (ED-if\xi)u=(E\xi+iDf)v,
\]
and further eliminating $u$ provides the following equation for $v$:
\[
  (D^2+f^2+\frac\xi E f')v =   (-\partial_y^2+f^2+\frac\xi E f')v=  (E^2-\xi^2)v.
\]
We verify from \eqref{eq:33} that the above equation still holds when $\xi=0$.

Let us assume that $E^2=\xi^2$.  Then $\xi=-E$ implies that the above equation for $v$ is $\fa^*\fa v=0$ whereas $\xi=E$ implies that $\fa\fa^*v=0$. The latter admits $v=0$ as a unique solution while the former admits the solution $v=e^{-F(y)}$. When $E=-\xi$, we then verify that
\[
   (\partial_y-f)u=-\fa^*u = \frac i\xi (\xi^2+f^2-f') v, \qquad -\fa^*\eta = \frac i\xi (\xi^2-f^2) v.
\]
These equations admit solution if and only if the right hand sides are orthogonal to the solutions in the kernel of $\fa$ by the Fredholm alternative. As a consequence, we observe that 
$
  \int_{\Rm} (\xi^2+f^2-f') v^2 dy = \int_{\Rm} (\xi^2-f^2) v^2 dy=0.
$
Only for specific $f$ are such constraints satisfied (they are for instance when $f=f_0\sgn{y}$). But in any case, they are satisfied for two values of $\xi$ at most and therefore only generate discrete spectrum in $H[f(y)]$ at most. Such discrete spectrum is irrelevant in our pursuit of branches of continuous spectrum and does not modify the value of Fredholm indices.

Therefore, $|E|\not=|\xi|$ when $v\not=0$. Once $v$ is found, then $u$ and $\eta$ are given by
\[
    u=\frac{1}{E^2-\xi^2} (\xi D + if E)v,\qquad \eta= \dfrac{1}{E}(\xi u+Dv).
\]
We are left with verifying that $u$ and $\eta$ thus defined are normalizable when $v$ is normalizable, which may depend on $f(y)$ but holds in the cases considered below.

We thus look for non-trivial solutions of the equation $ (-\partial_y^2+f^2+\frac\xi E f')v=  (E^2-\xi^2)v$. It does not seem possible to solve such an equation in closed form for arbitrary functions $f(y)$. We consider cases that admit closed-form expressions. 

We first consider the case $f(y)=f_0\sgn{y}$.  The solution $v$ is even in $y$ and given by $e^{-\mu y}$ for $y>0$ with $\mu>0$ such that $-\mu^2+f_0^2=E^2-\xi^2$. The jump conditions at $y=0$ read $v'(0-)-v'(0+)+\frac\xi E 2f_0=0$ or equivalently $\mu=-\frac\xi E f_0$. Plugging this into the above equation yields two solutions $E^2=f_0^2$ or $E^2=\xi^2$. We have already ruled out the latter. The former gives $E=-f_0\sgn{\xi}$ so that $\mu>0$. As a consequence, we observe that $E=\mp f_0$ is infinitely degenerate. The effect of the interaction of the bulk $0$ sheet and the varying $f$ is the presence of stationary waves since $E'(\xi)=0$. Moreover, we observe a spectral gap in $(-f_0,f_0)$ for the part of the spectrum corresponding to $v\not=0$. For energies $|E|<f_0$, we thus obtain only one branch of continuous spectrum given by the Kelvin waves. The number of interface modes is equal to $1$, not $2$ as (wrongly) predicted by the bulk interface correspondence. We refer the reader to \cite{tauber2019bulk,tauber2019anomalous} for other explicit calculations, obtained for instance in the presence of an odd viscosity contribution $\varepsilon$ so that $f$ in the equations is replaced by $f+\epsilon\Delta$.

Below, we analyze more general cases of $f(y)$ constant for $\pm y>0$. We find the surprising result that asymmetric edge states may appear when $0<f_-<f_+$ even though the bulk-difference invariant clearly vanishes in such a case.

First, we consider the practically relevant case $f(y)=\lambda y$ with $\lambda>0$ for concreteness. Now that $f'=\lambda$ is constant, we can relate the spectrum of the problem of interest to that of the quantum harmonic oscillator. We find that
$ \fa^*\fa v = (E^2-\xi^2-\lambda(1+\frac\xi E)) v$.
The spectrum of $\fa^*\fa$ is given by $2n\lambda$ so that 
\[
   E_n^2-\xi^2-\lambda \frac\xi E_n = (2n+1)\lambda,\qquad E_n^3 = (\lambda(2n+1)+\xi^2)E_n+\lambda\xi.
\]
We look for all solutions to the above equation, except for the case $E^2=\xi^2$, which we ruled out and occurs only when $n=0$. These are all cubic equations, which can be solved explicitly. The case $n=0$ provides two solutions, one with positive energy and the other one with negative energy. These are called  Yanai waves. We verify that $E_0'(\xi)>0$ for such waves and they therefore also propagate eastward.

The rest of the solutions are composed of Rossby waves and Poincar\'e waves coming in continuous branches of spectrum crossing an energy level different from $0$ a finite and even number of times so that they are topologically trivial (with a vanishing winding number); see \cite{delplace2017topological,tauber2019bulk} as well. 

As a summary of the above calculations, we therefore find that $2\pi\sigma_I$ is given by $+1$ when $f(y)=f_0\sgn{y}$ with $f_0>0$  while it is given by $+2$ when $f(y)=\lambda y$ with $\lambda>0$.
\\[2mm]
\noindent{\bf Coriolis force taking two values on half lines.}  We finally consider the case with $f(y)=f_+$ for $y>0$ and $f(y)=f_-<f_+$ for $y<0$. Then $v$, assuming $E\not=0$, solves the following system:
\[
  -\partial_y^2 v + f^2 v = (E^2-\xi^2)v, \quad y\not=0; \qquad 
  -(v'(0^+)-v'(0^-)) + \frac\xi E (f_+-f_-) v(0) =0.
\]
The solutions are given by $v(y)=e^{-\mu_+ y}$ for $y>0$ and $v(y)=e^{\mu_- y}$ for $y<0$ with $\mu_\pm>0$ as a necessary condition for normalization. We thus obtain the system of equations for $(E,\mu_+,\mu_-)$ knowing $(\xi,f_+,f_-)$:
\[
  E^2-\xi^2 = f_+^2-\mu_+^2= f_-^2-\mu_-^2, \qquad \mu_++\mu_-+\frac\xi E(f_+-f_-)=0.
\]
We want  solutions such that $\mu_\pm>0$. $f_+-f_->0$ implies that $\frac\xi E<0$.  We introduce:
\[
 f_e=\frac12(f_++f_-),\qquad f_o=\frac12(f_+-f_-) >0.
\]
From $\mu_+^2-\mu_-^2=(\mu_++\mu_-)(\mu_+-\mu_-)=f_+^2-f_-^2$ we deduce that
$
   \mu_+-\mu_-  + \frac E\xi (f_++f_-)=0.
$ 
Let us define $\nu=\frac E \xi$. Then,
\[
  \mu_+ + \nu^{-1} f_o + \nu f_e=0 \quad \mbox{ so that } \quad \mu_+^2 = \nu^{-2} f_o^2 + f_of_e + \nu^{2} f_e^2 = f_+^2+\xi^2(1-\nu^{2}).
\]
This gives an equation for $\nu$, or equivalently $E$, using $f_+^2-2f_of_e=f_o^2+f_e^2$,
\[
  (f_e^2+\xi^2)\nu^4 - (f_o^2+f_e^2+\xi^2)\nu^2 + f_o^2=0 \ \mbox{ or } \ (f_e^2+\xi^2) (\nu^2-1)(\nu^2-\frac{f_o^2}{f_e^2+\xi^2}) =0.
\]
We know the existence of a branch $E(\xi)=\xi$ (Kelvin waves with $v=0$) and ruled out $E(\xi)=-\xi$.
With $|E|\not=|\xi|$, we know that $\xi/E<0$ so that, since $f_o>0$, the only admissible solution is
\begin{equation}\label{eq:Exi}
  E = \frac{-\xi f_o}{\sqrt{f_e^2+\xi^2}}.
\end{equation}
For such a branch, we verify that $\partial_\xi E<0$. 
The only valid solutions of interest satisfy that $\mu_\pm=-\nu f_o \mp \nu^{-1} f_e>0$, or equivalently,
\[
    f_o > \nu^2 |f_e|\quad\mbox{ or } \quad \xi^2\geq  |f_e|(f_o-|f_e|),
\]
where we used the above expression for $\nu$. An admissible branch of continuous spectrum inside the bulk gaps implies $|E|<M:={\rm min}(|f_-|,|f_+|)=|f_o-|f_e||$. We thus find the constraints
\[
   \xi^2\geq |f_e|(f_o-|f_e|),\qquad \mbox{ and } \qquad (2f_o-|f_e|)\xi^2<|f_e|(f_o-|f_e|)^2.
\]
We then consider three cases. (i) $|f_e|\geq 2f_o$ in which case both constraints are satisfied; (ii) $2f_o>|f_e| \geq f_o$ in which case the first constraint is always satisfied and the second one provides
\[
  \xi^2  < |f_e| \frac{(|f_e|-f_0|)^2}{2f_o-|f_e|} .
\]
(iii) The third and last case is $|f_e|<f_o$ in which case the two constraints combined impose $2f_o-|f_e|<f_o-|f_e|$, which is not possible. 

We verify that $f_-<0<f_+$ belongs to case (iii). Only the Kelvin waves $E(\xi)=\xi$, as in the case $f_-=-f_+$, cross the bulk band gap. Surprisingly, as soon as $|f_e|\geq f_o$, which corresponds to $f_-< f_+\leq 0$ or $0\leq f_-<f_+$, then a branch of absolutely continuous spectrum crosses the intervals $(-M,0)$ and $(0,M)$ with $M=|f_e|-f_0$. Since $\partial_\xi E<0$, these branches have a winding number equal to $-1$.

We have thus considered the three different configurations: (i) when $f(y)=|\lambda|y$, then $I[H]=2$ as predicted by the bulk interface correspondence;  in the piece-wise constant case $f=f_-\chi_{y<0}+f_+\chi_{y>0}$ (ii) when $|f_e|<f_o$, then $I[H]=1$; and (iii) when $|f_e|>f_o$, then $I[H]=1-1=0$ as predicted by the bulk-interface correspondence, since $f_+$ and $f_-$ have the same sign.

\section{Equivalence of bulk-difference invariants}
\label{sec:appeq}
%

\begin{proof}[Lemma \ref{lem:equivCW}]
The objective is to recast $W_\alpha=\frac{i}{8\pi^2}\int Td^2k$ in terms of the projectors $\Pi_j$, where we use that $\partial_\omega G^{-1}=i$ and
\[
 T=T(k) :=\dint \tr G[\partial_1 G^{-1}G,\partial_2 G^{-1} G] d\omega .
\] 
For $H=\sum_i h_i \Pi_i$ so that $G=\sum_i(z-h_i)^{-1}\Pi_i$, we thus need to estimate
\[
 \dsum_{i,j,k} {\rm tr} \dint \frac{\Pi_i (\partial_1H \Pi_j\partial_2 H - \partial_2H \Pi_j\partial_1 H) \Pi_k}{(z-h_i)(z-h_j)(z-h_k)} d\omega = \dsum_{i,j} {\rm tr} \dint \frac{\Pi_i (\partial_1H \Pi_j\partial_2 H - \partial_2H \Pi_j\partial_1 H) }{(z-h_i)^2(z-h_j)} d\omega
\]
since $\Pi_i\Pi_k=\delta_{ik}\Pi_i$ and the trace is cyclic. We now evaluate the integrals over $\omega$, which have a different form depending on whether $i=j$ or not. For the latter case, we have
\[
  \frac{1}{(z-h_i)^2(z-h_j)}  = \frac{b_1}{\alpha-h_i+i\omega} + \frac{b_2}{(\alpha-h_i+i\omega)^2} + \frac{b_3}{\alpha-h_j+i\omega}
\]
with $b_3=-b_1=(h_j-h_i)^{-2}$ and $b_2=(h_i-h_j)^{-1}$. It remains to evaluate the integrals using that $\int(\alpha+i\omega)^{-1}d\omega=\pi\sgn{\alpha}$ to obtain 
\[
  \pi \frac{1}{(h_i-h_j)^2} (\sgn{\alpha-h_j}-\sgn{\alpha-h_i}) 
\]
since $\int(\alpha+i\omega)^{-k}d\omega=0$ for $k\geq2$ integer by residue calculation. When $i=j$, we similarly obtain a vanishing contribution. Thus,
\[
T= \dsum_{i\not=j}  \pi (\sgn{\alpha-h_j}-\sgn{\alpha-h_i}) {\rm tr} \frac{\Pi_i (\partial_1H \Pi_j\partial_2 H - \partial_2H \Pi_j\partial_1 H) }{(h_i-h_j)^2}. 
\]
By cyclicity of the trace, we obtain
\begin{equation}\label{eq:w3trace}
 T= \dsum_{i<j}  2\pi (\sgn{\alpha-h_j}-\sgn{\alpha-h_i}) {\rm tr} \frac{\Pi_i (\partial_1H \Pi_j\partial_2 H - \partial_2H \Pi_j\partial_1 H) }{(h_i-h_j)^2}. 
\end{equation}
The above formula is convenient for explicit computations of the invariants. When the projectors $\Pi_i=\psi_i\otimes\psi_i$ are rank one, using that $\aver{\psi_i,\gamma\psi_j}=\overline{\aver{\psi_j,\gamma\psi_i}}$ for any Hermitian matrix $\gamma$, we observe that
\begin{equation}\label{eq:w3trace2}
 T= 2\pi \dsum_{i<j}  \dfrac{\sgn{\alpha-h_j}-\sgn{\alpha-h_i}}{(h_i-h_j)^2} 
2i \Im (\aver{\psi_i,\partial_1 H \psi_j}\aver{\psi_j,\partial_2 H \psi_i}).
\end{equation}

From the above expression for $H$, we deduce that $\partial_kH=\sum_l \partial_k h_l \Pi_l + h_l\partial_k \Pi_l$. For $i\not=j$, all terms involving $\partial_k h_l$ cancel as $\Pi_i\Pi_j\Pi_k=0$ so that
\[
  \dsum_{i<j}  2\pi (\sgn{\alpha-h_j}-\sgn{\alpha-h_i}) \sum_{k,l}\dfrac{h_kh_l}{(h_i-h_j)^2} 
  {\rm tr} \Pi_i (\partial_1\Pi_k \Pi_j\partial_2 \Pi_l - \partial_2\Pi_k \Pi_j\partial_1 \Pi_l) . 
\]
$\Pi_m\Pi_n=\delta_{mn}\Pi_n$ implies $\partial_m\Pi_k \Pi_l + \Pi_k \partial_m \Pi_l = \delta_{kl}\partial_m\Pi_k$ so that $(k,l)$ in the above formula have to equal $i$ or $j$. The sum over $(k,l)$ is thus
\[
   {\rm tr} \dfrac{h_i^2 \Pi_i \partial_1 \Pi_i \Pi_j \partial_2 \Pi_i + h_ih_j (\Pi_i\partial_1 \Pi_i \Pi_j \partial_2 \Pi_j + \Pi_i\partial_1 \Pi_j \Pi_j \partial_2 \Pi_i) +h_j^2 \Pi_i \partial_1 \Pi_j \Pi_j \partial_2 \Pi_j  }{(h_i-h_j)^2}
\]
minus the contribution exchanging the order of the derivatives. The above term is
\[
  {\rm tr}  \Pi_i \partial_1 \Pi_i \Pi_j \partial_2 \Pi_i = -{\rm tr} \Pi_i\partial_1 \Pi_i \Pi_j \partial_2 \Pi_j = {\rm tr}  \Pi_i \partial_1 \Pi_j \Pi_j \partial_2 \Pi_j =- \rm{ tr} \Pi_i\partial_1 \Pi_j \Pi_j \partial_2 \Pi_i
\]
since $\Pi_j^2=\Pi_j$ so that $\partial_k\Pi_j=\partial_k\Pi_j\Pi_j+\Pi_j\partial_k\Pi_j$  and ${\rm tr}\Pi_iA={\rm tr} A\Pi_i$.
We thus have 
\[
 T = 2\pi \dsum_{i<j}  (\sgn{\alpha-h_j}-\sgn{\alpha-h_i}) {\rm tr} \Pi_i (\partial_1\Pi_i\Pi_j\partial_2\Pi_i-\partial_2\Pi_i\Pi_j\partial_1\Pi_i).
\]
Therefore, 
\[
 T= -4\pi \dsum_{h_i<\alpha} \dsum_{h_j>\alpha} {\rm tr} \Pi_i (\partial_1\Pi_i \Pi_j \partial_2\Pi_i-\partial_2\Pi_i \Pi_j \partial_1\Pi_i),
\]
or using ${\rm tr} \Pi_i\partial_1\Pi_i\Pi_j \partial_2\Pi_i={\rm tr} \Pi_j \partial_2\Pi_i \Pi_i\partial_1\Pi_i$
\[
  T = 4\pi  \dsum_{h_i<\alpha}  {\rm tr}  (\dsum_{h_j>\alpha}  \Pi_j) (\partial_1\Pi_i \Pi_i \partial_2\Pi_i-\partial_2\Pi_i \Pi_i \partial_1\Pi_i).
\]
This is also, for $i_0$ the index so that $h_{i_0}<\alpha< h_{i_0+1}$,
\[
  T = 4\pi  \dsum_{i\leq i_o}  {\rm tr}  (I-\dsum_{j\leq i_0}  \Pi_j) (\partial_1\Pi_i \Pi_i \partial_2\Pi_i-\partial_2\Pi_i \Pi_i \partial_1\Pi_i).
\]
Now from
$
 {\rm tr} \Pi_j \partial_1 \Pi_i \Pi_i \partial_2 \Pi_i = {\rm tr} \Pi_i \partial_2 \Pi_j \Pi_j \partial_1 \Pi_j
$
for $i\not=j$ and ${\rm tr} [\Pi_i\partial_1\Pi_i,\Pi_i\partial_2\Pi_i]=0$, we deduce that the double sum for $(i,j)\leq i_0$ vanishes and
\[
  T = 4\pi  \dsum_{i\leq i_0}  {\rm tr}  (\partial_1\Pi_i \Pi_i \partial_2\Pi_i-\partial_2\Pi_i \Pi_i \partial_1\Pi_i) = -4\pi \dsum_{i\leq i_0}  {\rm tr}\Pi_i [\partial_1\Pi_i,\partial_2\Pi_i].
\]
This shows that for $\alpha$ in the $i_0$th gap, the winding number is given by the sum of the first $i_0$ Chern numbers as expected.
\end{proof}

\medskip\noindent
{\bf Application to bulk invariant calculations for the $3\times3$ system \eqref{eq:s33}.} We now use \eqref{eq:w3trace2} to compute the bulk-difference $W_\alpha$ invariant. We recast the system \eqref{eq:s33} as $H=\xi\gamma_1+\zeta\gamma_4-f\gamma_7$. This is diagonalized as follows. We find $I=\Pi_-+\Pi_0+\Pi_+$ for $H=\kappa(\Pi_+-\Pi_-)$ and $\Pi_j=\psi_j \psi_j^*$ rank one projectors given by
\begin{equation}\label{eq:eigen3}
  \psi_0 = \frac1\kappa \left(\begin{matrix}  if \\ \zeta \\ -\xi \end{matrix} \right) ,\quad
  \psi_\pm = \frac{1}\rho \left(\begin{matrix}  if\xi \pm \kappa\zeta \\ \xi\zeta \pm if\kappa \\ \zeta^2+f^2\end{matrix} \right)
\end{equation}
and $(\xi,\zeta)-$dependent eigenvalues given by
\[
   h_0=0,\qquad h_\pm = \pm \kappa, \qquad \mbox{ for } \quad  \kappa=\sqrt{\xi^2+\zeta^2+f^2}\quad\mbox{ and } \quad  \rho=\kappa\sqrt{2(f^2+\zeta^2)}.
\]
We find $\partial_\xi H=-\partial_\xi G^{-1}=\gamma_1$ and $\partial_\zeta H=-\partial_\zeta G^{-1}=\gamma_4$. As a consequence, \eqref{eq:w3trace2} gives $W_\alpha=\frac{i}{8\pi^2}\int_{\Rm^2} T d\xi d\zeta$ with
\[
  T= 2\pi \dsum_{1\leq i <j \leq 3}  \dfrac{\sgn{\alpha-h_j}-\sgn{\alpha-h_i}}{(h_i-h_j)^2} 
2i\Im (\aver{\psi_i,\gamma_1 \psi_j}\aver{\psi_j,\gamma_4 \psi_i}).
\]
 For $\alpha>0$, the only differences of signs contributing non trivial terms are $(i,j)$ with $i<j$ in $\{(-,+),(0,+)\}$ with then  $\sgn{\alpha-h_j}-\sgn{\alpha-h_i}=-2$. We find
\[
  \aver{\psi_+,\gamma_4\psi_-}=0,\quad  \aver{\psi_+,\gamma_4\psi_0}=\frac1\rho(if \kappa-\zeta\xi),\quad
\aver{\psi_0,\gamma_1\psi_+} =\frac{1}{\rho}(\zeta^2+f^2).
\]
Thus, $2i \Im  (\aver{\psi_0,\gamma_1 \psi_+}\aver{\psi_+,\gamma_4 \psi_0}) = \frac{if}{\kappa}$.
Note that $h_+-h_0=\kappa$. The bulk difference is defined for $f_\pm=\pm f$ on the top and bottom planes, which yields
\[
  -W_\alpha = 2\sgn{f} \frac{i}{8\pi^2} \dint_{\Rm^2} \frac{i|f|}{\kappa} \frac{-4\pi}{\kappa^2} d\xi d\zeta = \frac{\sgn{f}}{\pi} \dint_{\Rm^2} \frac{|f|d\xi d\zeta}{(\zeta^2+\xi^2+f^2)^{\frac32}}  = 2\sgn{f},
\]
as a standard integration for two-dimensional invariants. This was obtained for $\alpha>0$ and $f_-=-f_+$. All other expressions are obtained similarly. In particular, we observe that $W_\alpha$ after some algebra takes the same value for $\alpha<0$ so that $c_+=c_++c_0=-c_-=-2\sgn{f}$, which yields \eqref{eq:chern3}. See \cite{souslov2019topological,tauber2019bulk} and in particular \cite[Appendix A]{graf2020topology} for alternative methods to compute the above invariants.

\section{Bulk-interface correspondences}
\label{sec:bic}

This paper presents two methods to compute the interface conductivity for general classes of operators $H_V$. Both involve two steps. The first step is to assume that $H_V=H[\mu(y)]+V$ with $V$ a perturbation in an appropriate sense so that we may prove $\sigma_I[H_V]=\sigma_I[H]$. The second step is the computation of $\sigma_I$ for an operator that is invariant by translations along the $x$ axis. In the Fourier domain $x\to\xi$, we thus end up with the analysis of problems of the form  $H(\xi)\psi(\xi)=E(\xi)\psi(\xi)$ for $E(\xi)$ an energy inside the bulk band gap.

The two methods presented here in sections \ref{sec:ii} and \ref{sec:BI}  are related to the notion of spectral flow and the similar one of spectral asymmetry.  The computation given in section \ref{sec:ii} is based on the assumption of finitely many spectral branches $\xi\to E(\xi)$. The computations in section \ref{sec:BI} are more closely related to (a regularization of) a spectral asymmetry. 

In  this appendix, we briefly present some related results on the bulk-interface correspondences primarily developed in \cite{fukui2012bulk,essin2011bulk,volovik2009universe}. All revolve around the following observation. Consider the graph $\xi\mapsto E(\xi)$ and assume that it crosses the value $0$ a finite number of times. If that number is even, then the graph can be continuously deformed (with fixed end points) so that it no longer crosses the line $0$. Such a deformation is not possible when that number is odd, where the minimal number of crossings is $\pm1$ accounting for direction of crossing. This number is topologically protected and equivalent to the winding number of the spectral branches we considered in section \ref{sec:ii}. This asymmetric edge transport is also the main  physically observable manifestation of the non-trivial topology of the system.

\medskip
\noindent{\bf Spectral flow.}
The above asymmetry may be formalized by the notion of spectral flow. Consider the $2\times2$ system and assume that $m(y)$ is a continuous function equal to $m_-$ for $y<-|y_0|<0$ and equal to $m_+$ for $y>|y_0|$. The bulk Hamiltonian has a spectral gap in $(-m_0,m_0)$ for $m_0$ the minimum of $|m_-|$ and $|m_+|$.
For any energy $E$ in the spectral gap $(-m_0,m_0)$, we look at the graph $\xi\mapsto(H[\xi]-E)$ for eigenvalues of  $H[\xi]$ in the spectral gap as well. The spectrum of $H[\xi]$ restricted to the spectral gap is only composed of point spectrum as already recalled in the preceding section. 
 We then consider an energy $-m_0<E<m_0$ and define the number of protected edge states as \cite{fukui2012bulk}
\[
   I(E) ={\rm SF}[H[\xi]-E],
\]
the number of signed crossings of $H[\xi]$ through $E$ (counted as positive from bottom to top and negative from top to bottom).  The method devised in  \cite{fukui2012bulk} to compute the spectral flow is to relate it (semi-formally) to the index of the Fredholm operator :
\begin{equation}\label{eq:Hv}
   H_v = v\partial_\xi + H[\xi] -E = \left( \begin{matrix}  \fa_\xi-E & -i\fa_y \\ i\fa_y^* & -\fa_\xi^*-E \end{matrix} \right), \qquad \fa_y=\partial_y + m(y) ,\quad \fa_\xi = v \partial_\xi + \xi
\end{equation}
for some $v>0$ arbitrary. Using that the kernel of  $\fa_y$ is non-trivial whereas that of $\fa_y^*$ is trivial when $m_-<0<m_+$ and that the kernel of $\fa_\xi-E$ is non-trivial whereas that of $\fa^*_\xi-E$ is, it is straightforward to obtain that 
$
   I[H_v]  =  - \frac12 (\sgn{m_+} - \sgn{m_-}),
$
which is also the bulk-difference invariant $c_+=-c_-$ in \eqref{eq:chern2}. This is the bulk-interface correspondence for the $2\times2$ system. The correspondence  only depends on the asymptotic behavior of $m(y)$ for large $|y|$. This shows that $\fa_\xi$ and $\fa_y$ can be continuously modified without changing the invariant so long as $m_\pm$ remain of constant sign. 

The above calculations extend to a large class of models as shown in \cite{fukui2012bulk}. However, the calculations we  obtained earlier show that the above correspondence cannot extend to arbitrary $3\times3$ models since the spectral flow depends on some details of the coefficient $f(y)$. We will revisit the above notion of spectral flow below.

\medskip
\noindent
{\bf Spectral asymmetry.}
An intuitive picture of the bulk-interface correspondence was developed in \cite[Chapter 22]{volovik2009universe} and further analyzed in \cite{essin2011bulk,gurarie2011single}. It is based on the spectral asymmetry, 
\[
   \nu(\xi;\alpha) = {\rm Tr} \dint_{\Rm} \dfrac{d\omega}{2\pi i} {\mathcal G} \partial_\omega {\mathcal G}^{-1}  = {\rm Tr} \dint_{\Rm} \dfrac{d\omega}{2\pi } {\mathcal G} 
\]
defined for $\alpha\in\Rm$, where ${\mathcal G}=(\alpha+ i\omega-H[\xi])^{-1}$ is the Green's (resolvent) operator, whose kernel is the Green's function, and where ${\rm Tr}$ is the operator trace. The above operator is in fact not necessarily trace-class and the asymmetry will be appropriately regularized below. Heuristically, we expect $\nu(\xi;\alpha)$ to provide the difference of the numbers of eigenvalues of $H$ that are larger and smaller than $\alpha$ since, formally,
\[
   \nu(\xi;\alpha) = -{\rm Tr} \dint \dfrac{d\omega}{2\pi} \dfrac{H[\xi]-\alpha}{\omega^2 + (H[\xi]-\alpha)^2} = -\frac12\dsum_n \sgn{E_n(\xi)-\alpha}.
\]
The sum in $\nu(\xi_2;\alpha)-\nu(\xi_1,\alpha)$ counts the number of eigenvalues $E_n(\xi)$ that cross $\alpha$ as $\xi$ runs from $\xi_1$ to $\xi_2$ as does the spectral flow considered earlier.

The above sign function is defined only when $\alpha$ is not an eigenvalue of $H[\xi]$. Assuming that this is the case when $|\xi|$ is large, then, still heuristically, we expect that
\begin{equation}\label{eq:heursf}
 -I(\alpha) = \lim\limits_{\xi\to\infty} \big( \nu(\xi; \alpha)  - \nu(-\xi; \alpha) \big).
\end{equation}

The works in \cite{essin2011bulk,volovik2009universe} then recast the right-hand side in \eqref{eq:heursf} as an integral of the symbol of $H[\mu(y)]$, or equivalently the symbol of the resolvent (Green's) operator. Such a derivation, called gradient expansion in these references, will be justified below by using semiclassical calculus in a spirit similar to the derivation of the Fedosov-H\"ormander index theorem \cite[Theorem 19.3.1]{H-III-SP-94}.  We also refer to  \cite{drouot2019microlocal} for a recent application of semiclassical calculus for the derivation of the bulk-interface correspondence for periodic magnetic Schr\"odinger operators.  

In the rest of the section, we introduce this calculus, which only applies for Hamiltonians with smooth coefficients, and recast the bulk-interface correspondence as a calculation of an index for an appropriate Fredholm operator that is similar to \eqref{eq:Hv}. In doing so, we obtain that \eqref{eq:heursf} actually holds only in specific cases, which for the $2\times2$ and $3\times3$ systems correspond to the cases where $\mu(y)$ converges to infinity at infinity.

\section{Semiclassical calculus and Helffer-Sj\"ostrand formula}
\label{sec:hpdo}
We collect here notation and results on pseudo-differential, semiclassical, and spectral calculus used in this paper following \cite{bolte2004semiclassical,davies_1995,dimassi1999spectral,zworski2012semiclassical}, to which we refer for details.

Let $V=\Rm^d$ and $V'\simeq\Rm^d$ its dual. A bounded matrix-valued operator $A$ from ${\mathcal S}(V)\otimes \Cm^n$ the Schwartz space to its dual ${\mathcal S}'(V)\otimes \Cm^n$ admits a (Schwartz) distribution kernel $K_A\in {\mathcal S'}(V\times V)\otimes \Mm_n(\Cm)$ and can be represented for $0<h\leq1$ as
\begin{equation}\label{eq:hpdo}
  A = {\rm Op}^w_h(a) := {\rm Op}_h(a),\quad   {\rm Op}_h(a)\psi(x) = \dfrac{1}{(2\pi h)^d} \dint_{V'\times V} e^{i\frac{x-y}h\cdot\xi} a(\frac{x+y}2,\xi) \psi(y) dy d\xi,
\end{equation}
where $a(x,\xi)\in {\mathcal S}'(V\times V') \otimes \Mm_n$ is the Weyl symbol of $A$ defined as 
\[
  a(x,\xi) = \dint_V e^{-i\frac{y\cdot\xi}h} K_A(x+\frac y2,x-\frac y2) dy.
\]
The notation $a^w(x,hD):={\rm Op}_h(a)$ is also used to define the semiclassical operator (\hpdo) for the Weyl quantization of $a$ with $D=\frac1i\nabla$. 
 Finally, we use the notation ${\rm Op}^w(a)={\rm Op}(a):={\rm Op}_1(a)$ for the Weyl quantization of pseudo-differential operators defined at the scale $h=1$. 

\medskip

To define operators with smoother kernels that can be composed with each other, we define the space of {\em order functions} $m(x,\xi)$ from $V\times V'\to[0,\infty)$ satisfying the growth condition:
\begin{equation}\label{eq:order}
  m(x,\xi) \leq C (1+|x-y|+|\xi-\zeta|)^N m(y,\zeta)
\end{equation}
for some constants $C=C(m)$ and $N=N(m)$. Examples of interest are $(1+|x|^2+|\xi|^2)^s$ for $s\in\Rm$ as well as ${\rm max}(0,\pm x_1)$.
We then denote by $S(m)$ the Fr\'echet space of symbols $a(x,\xi)\in C^\infty(V\times V') \otimes \Mm_n$  such that 
\begin{equation}\label{eq:S}
 \|\partial^\alpha_x\partial^\beta_\xi a(x,\xi) \|_{\Mm_n} \leq C_{\alpha,\beta} m(x,\xi).
\end{equation}
For $h-$dependent symbols $a(\cdot;h)$, we say that $a\in S^0(m)$ when $a(\cdot;h)\in S(m)$ uniformly in $0<h\leq1$.

For two operators $a^w$ and $b^w$ with symbols $a\in S^0(m_1)$ and $b\in S^0(m_2)$, we then define the composition $c^w=a^w\circ b^w$ with symbol $c\in S^0(m_1m_2)$ given by the Moyal product
\begin{equation}\label{eq:sharph}
c(x,\xi) =(a\sharp_h b) (x,\xi) := \Big( e^{i\frac h2(\partial_x\cdot\partial_\zeta - \partial_y\cdot\partial_\xi)} a(x,\xi) b(y,\zeta)\Big)_{|y=x;\zeta=\xi}.
\end{equation}
In other words, ${\rm Op}_h(a) {\rm Op}_h(b) = {\rm Op}_h(a\sharp_h b)$. We also define $\sharp=\sharp_1$ when $h=1$.

For $a\in S^0(1)$, we obtain (\cite[Th. 7.11]{dimassi1999spectral},\cite[Prop. 1.4]{bolte2004semiclassical}) that ${\rm Op}_h(a)$ is bounded as an operator in ${\mathcal L}(L^2(V)\otimes \Cm^n)$ with bound uniform in $0<h\leq1$ so that $I-h{\rm Op}_h(a)$ is invertible on that space when $h$ is sufficiently small.

An operator is said to be (semiclassically) elliptic when the symbol $a=a(x,\xi;h)\in S^0(m)$ is invertible in $\Mm_n$ for all $(x,\xi)\in V\times V'$ with $a^{-1}\in S^0(m^{-1})$. For $\sigma$ a Hermitian-valued symbol and $H_h={\rm Op}_h(\sigma)$, we then obtain that $z-H_h$ is elliptic for $z=\lambda+i\omega$ with $\omega\not=0$. Moreover, $(z-H_h)^{-1}$ is a operator with symbol $r_z \in S^0(1)$ uniformly in $0<h\leq h_0$ sufficiently small, $|\lambda|$ bounded and $|\omega|$ bounded away from $0$; see \cite[Prop. 8.6]{dimassi1999spectral}, which extends to the matrix-valued case, and Lemma \ref{lem:rz} for a more precise bound for $\omega$ small.

\medskip

We next recall the sharp semiclassical G\aa rding inequality \cite[Theorem 7.12]{dimassi1999spectral} (which extends to the vectorial case without modification) stating that for $a$ a Hermitian-valued symbol in $S^0(1)$ with eigenvalues greater than or equal to $\beta\in\Rm$ for all $(x,\xi)\in V\times V'$, then 
\begin{equation}\label{eq:garding}
  ({\rm Op}_h(a)\psi,\psi) \geq (\beta-Ch) \|\psi\|^2
\end{equation}
for all $\psi\in L^2(V)\otimes\Cm^n$ and for $C$ a constant independent of $0<h\leq h_0$ {\em sufficiently small}. Thus, for $\beta>0$ and $h$ sufficiently small, the operator ${\rm Op}_h(a)$ has spectrum bounded away from $0$.

%

\medskip

Finally, we recall some results on spectral calculus and the Helffer-Sj\"ostrand formula following \cite{davies_1995,dimassi1999spectral}; see also \cite{bolte2004semiclassical} for the vectorial case. For any self-adjoint operator $H$ from its domain ${\mathcal D}(H)$ to $L^2(V)\otimes\Cm^n$ and any bounded continuous function $\varphi$ on $\Rm$, then $\varphi(H)$ is uniquely defined as a bounded operator on $L^2(V)\otimes\Cm^n$ \cite[Chapter 4]{dimassi1999spectral}. Moreover, for $\varphi$ compactly supported, we have the following spectral representation
\begin{equation}\label{eq:hs}
  \varphi(H) = -\frac1\pi \dint_{\Cm} \bar\partial \tilde\varphi(z) (z-H)^{-1} d^2z,
\end{equation}
where, for $z=\lambda+i\omega$, $d^2z:=d\lambda d\omega$, $\bar\partial=\frac12\partial_\lambda+\frac1 2\partial_\omega$, and where $\tilde\varphi(z)$ is an almost analytic extension of $\varphi$.
The extension $\tilde\varphi$ is compactly supported in $\Cm$. 
Moreover, $\tilde\varphi(\lambda+i0)=\varphi(\lambda)$ and $\bpar \tilde\varphi(\lambda+i0)=0$, whence the name of {\em almost} analytic extension.  We can choose the almost analytic extension such that $|\bar\partial\tilde\varphi| \leq C_n |\omega|^n$  for any $n\in\Nm$ in the vicinity of the real axis uniformly in $(\alpha,\omega)$ on compact sets.    Several explicit expressions, which we do not need here, for such extensions are available in \cite{davies_1995,dimassi1999spectral}.


\end{document}